\documentclass[twoside]{article}

\usepackage[accepted]{aistats2025}
\usepackage{preamble}

\begin{document}

%
\runningtitle{Two-Stage Rollout Designs with Clustering for Interference}

%

\twocolumn[

\aistatstitle{Analysis of Two-Stage Rollout Designs with Clustering \\for Causal Inference under Network Interference}

\aistatsauthor{ Mayleen Cortez-Rodriguez \And Matthew Eichhorn \And  Christina Lee Yu }

\aistatsaddress{Cornell University \And Cornell University \And Cornell University} ]

\begin{abstract}
  Estimating causal effects under interference is pertinent to many real-world settings. Recent work with low-order potential outcomes models uses a rollout design to obtain unbiased estimators that require no interference network information. However, the required extrapolation can lead to prohibitively high variance. To address this, we propose a two-stage experiment that selects a sub-population in the first stage and restricts treatment rollout to this sub-population in the second stage. 
We explore the role of clustering in the first stage by analyzing the bias and variance of a polynomial interpolation-style estimator under this experimental design. Bias increases with the number of edges cut in the clustering of the interference network, but variance depends on qualities of the clustering that relate to homophily and covariate balance. There is a tension between clustering objectives that minimize the number of cut edges versus those that maximize covariate balance across clusters. 
Through simulations, we explore {a bias-variance} trade-off and compare the performance of the estimator under different clustering strategies.
\end{abstract}

\section{Introduction} \label{sec:introduction}

The stable unit treatment value assumption (SUTVA) is critical for many classic causal inference methods, but is violated in settings with interference, where the outcome of an individual can be affected by the treatment assignment of another. Interference introduces bias into estimators, potentially leading to inaccurate conclusions about causal effects when ignored \citep{sobel2006randomized}. Many domains experience interference. In evaluating the effect of a public health intervention such as a vaccine, peer effects from herd immunity play a role \citep{hudgens2008toward}. In evaluating the effect of a new feature on a social media platform, a user's engagement is affected by the behaviors of their social connections, whether they are directly exposed to the new feature or not \citep{biswas2018estimating, aral2012identifying}.

We exploit both a rich two-stage randomized experimental design and a flexible potential outcomes model to estimate the \textit{total treatment effect} (TTE), the difference in average outcomes of the population when everyone is treated versus untreated. The TTE estimand, sometimes called the \textit{global average treatment effect}, is a natural choice in applications where the decision-maker needs to decide between adopting the new intervention for everyone or sticking with the status quo, such as a tech company choosing a single user-feed content recommendation algorithm for their social media platform. The class of experimental designs we consider in this paper are called \textit{staggered rollout designs}, where treatment is assigned over different time periods to increasing subsets of participants until it has been rolled out to all subjects designated for treatment. This style of experiment is common on online platforms \citep{xu2018sqr} to mitigate the possible, unknown risks related to introducing a new feature, and in medicine \citep{brown2006stepped}, where for logistical or financial reasons it may be impossible to deliver treatment to all participants at once. 

\paragraph{Related Work.} 
Many prior approaches for causal inference under interference consider cluster randomized designs \citep{sobel2006randomized,hudgens2008toward,liu2014large,ugander2013graph, gui2015network, eckles2017design, auerbach2021local,brennan2022cluster, ugander2023randomized}. These works exploit structural assumptions on the underlying interference network to reduce bias in the difference in means estimator or reduce variance in the Horvitz-Thompson estimator via cluster randomized designs. Some of these works rely on the assumption of \textit{partial interference}, which posits that the underlying network is made up of disjoint groups and interference only occurs within, not across, groups \citep{sobel2006randomized,hudgens2008toward, liu2014large, bhattacharya2020causal, auerbach2021local}. Other works are devoted to proposing cluster randomized designs that exploit knowledge about the network and its structure to minimize the number of edges that cross clusters in settings where partial interference may not hold \citep{ugander2013graph, gui2015network, eckles2017design, brennan2022cluster, ugander2023randomized}. 

Another strand of literature exploits assumptions on the potential outcomes in their methodology while considering simpler, unit-randomized designs instead of, or in addition to, cluster-randomized designs \citep{toulis2013estimation,cai2015social, gui2015network, parker2017optimal,chin2019featureengineering}. These approaches assume linear or generalized linear potential outcomes models, reducing the estimation task to regression. A drawback of these approaches is their assumption of \textit{anonymous interference}, which posits that only the number, not identity, of treated units affects an individual's outcome \citep{hudgens2008toward, liu2014large, li2022random}, and the entire population shares the same outcomes model.

To address this, \citet{cortez2022neurips} introduced a flexible class of potential outcomes models that allow for heterogeneous treatment effects by relaxing the anonymous interference assumption and instead imposing the $\beta$-\textit{order interactions} assumption, where interference effects are constrained to small subsets of the population. Other recent work has also adopted this model \citep{cortez2023jci, lowordercluster-2024}.
In both our work and \cite{cortez2022neurips}, the class of estimators under consideration are based on polynomial interpolation, where a key insight is that under a $\beta$-order potential outcomes model, the expected average outcome of the population is a $\beta$-degree polynomial in the treatment level. This style of estimator was first introduced by \citet{yu2022PNAS} for the case when the polynomial has degree $1$, and generalized by \citet{cortez2022neurips} to polynomials of higher degree. A drawback is that the estimator has prohibitively high variance when the polynomial degree is greater than $1$, especially when the treatment probability is small. The present work addresses this by using a two-stage experimental design to reduce variance in polynomial interpolation estimators for the TTE. A key contribution of \cite{yu2022PNAS} and \cite{cortez2022neurips} is the unbiased estimation of causal effects \textit{without any knowledge of the underlying interference network}, which the majority of prior approaches require. Our approach does not require network knowledge, but we show how using graph knowledge to select a good graph clustering with which to correlate treatments may improve our estimator's performance.

A few recent works also study rollout designs for causal inference under interference. \citet{han2022detecting} present statistical tests to detect the presence of interference using rollout designs. Our work differs from theirs in that our focus is estimating treatment effects under interference, not detecting its presence. \citet{ari2023rollout} leverage rollout designs as part of a model selection mechanism to select a ``best'' model for interference. While they also study the TTE, their focus is on its identification conditions and how rollout designs aid in satisfying them. \citet{viviano2020experimental} does not leverage rollout designs, instead considering a two-wave experiment that uses a pilot study in the first wave to minimize variance in causal effect estimation from a main experiment in the second wave. Unlike our two-stage approach, their two-wave design does not utilize a staggered rollout. Furthermore, \citet{viviano2020experimental} requires anonymous interference, whereas our approach does not. \cydelete{The approach most similar to our work is that of \citet{cortez2022neurips}, who consider staggered rollout designs for estimating the TTE alongside the class of models and estimators we analyze in the present work. We present a modification to their experiment design that reduces the mean squared error (MSE) in many settings.}

\paragraph{Contributions}
We propose a two-stage experiment design to address the high variance of polynomial interpolation estimators under $\beta$-order interactions with large $\beta$. {Given an overall treatment budget $p$, for a chosen parameter $q \in [p,1]$}, the first stage samples a {$p/q$ fraction} subset of the population, and the second stage runs a staggered rollout on the selected subset {with an effective budget of $q$.} \cydelete{We introduce a parameter $q$ that helps us consider a trade-off between the sampling in the first stage and the rollout of treatment in the second stage. Larger values of $q$ correspond to selecting a smaller subset in the first stage and treating a higher proportion of these selected units in the second stage.}
{We propose a polynomial interpolation estimator that uses the higher effective budget $q$, and scales the final outcome to account for the fact that only $p/q$ fraction of units are eligible for treatment in stage two. The increased effective budget reduces the variance from polynomial interpolation. We show the following insights:}
\begin{itemize}
    \item This two-stage estimator has less variance than a one-stage rollout interpolation estimator, but the {sub-}sampling in the first stage introduces bias.
    \item Larger values of $q$ lead to higher bias due to edges cut in stage one of the design (that is, edges crossing between selected units and unselected units).
    \item When clustering is used in the first stage, the bias and variance of the estimator are affected by the edges between clusters and the variance of the average treatment effect across clusters.
    \item Since the variance of average cluster effects relates to homophily, we see a tension between two clustering objectives: minimizing cut edges versus maximizing covariate balance. 
    \item Even without network or covariate information, the two-stage approach improves for large values of $\beta$ (i.e. richer models); a good clustering can help improve further.
\end{itemize}

\section{Preliminaries} \label{sec:preliminaries}
We estimate the effect of a treatment on a population of $n$ individuals, denoted $[n] := \{ 1, \hdots, n \}$, via a randomized experiment. Their treatment assignments are collected in a binary vector $\bz \in \{0,1\}^n$, where $z_i=1$ (resp. $z_i=0$) indicates that unit $i$ is assigned to treatment (resp. control).
We allow for \textit{interference}, so the potential outcome of individual $i$ may be a function of the entire treatment vector $Y_i:\{0,1\}^n\to \mathbb{R}$. We frame our analysis around potential outcomes with the following two features.

First, individual $i$'s outcome is a function of a small subset of the population. We visualize this subset as $i$'s in-neighborhood $\cN_i$ in a directed interference graph, where an edge from $j$ to $i$ indicates that $j$'s treatment affects $i$'s outcome; we call $j$ an \textit{in-neighbor} of $i$. We assume the graph does not change over the timescale of the experiment.

\textbf{Assumption 1 (Neighborhood Interference):} 

If $\bz, \bz'$ have $z_j\!=\!z'_j \ \forall\; j\!\in\!\cN_i$, then $Y_i(\bz)\!=\!Y_i(\bz') \hfill \forall \; i$. 

We use $d := \max_i |\cN_i|$ to denote the maximum in-degree of the network. Differing from most prior work with neighborhood interference, the underlying interference graph may be \textit{unknown}. In some cases, we leverage various levels of network knowledge such as covariate information or full edge information. 

Second, following \citet{cortez2022neurips}, we use the binary nature of the treatments $z_i \in \{0,1\}$ to represent \textit{any} potential outcomes under neighborhood interference as a polynomial in $\bz$:
\begin{equation}
    Y_i(\bz) = \textstyle\sum_{\cS \subseteq \cN_i} c_{i,\cS} \textstyle\prod_{j \in \cS} z_j, \label{eq:ppom}
\end{equation}
where the coefficients $c_{i,\cS}$ represent the additive effect to individual $i$'s outcome if everyone in $\cS$ is treated. {Our second assumption posits that each outcome is affected only by \textit{small} treated subsets.}

\textbf{Assumption 2 ($\beta$-Order Interactions):} $c_{i,\cS} = 0$ for all $|\cS| > \beta$. 

Under this assumption, $Y_i(\bz)$ is a polynomial with degree at most $\beta$. This is motivated by settings where an individual is separately affected by smaller sub-communities of their neighbors (e.g. colleagues, family members, close friends) rather than the neighborhood as a whole. The case $\beta=1$ is the heterogeneous linear outcomes model explored in \cite{yu2022PNAS}, which generalizes the linear models commonly used in applied settings. When $\beta = d$, we return to the unrestricted neighborhood interference setting. \medelete{Like \cite{cortez2022neurips,cortez2023jci}, the current work exploits this structure when $\beta$ is much smaller than $d$.}

Our estimand of interest is the \textit{total treatment effect} (TTE), the average difference in outcomes when everyone versus no one is treated:
\begin{equation} \label{eq:tte_coeff}
    \tfrac{1}{n} \textstyle\sum_{i=1}^{n} \big( Y_i(\mathbf{1}) - Y_i(\mathbf{0}) \big) = \tfrac{1}{n}\textstyle\sum_{i\in[n]}\textstyle\sum_{\cS \in \cS_i^\beta \setminus \varnothing} c_{i,\cS},
\end{equation}
where $\cS_i^\beta :=\{\cS \subseteq \cN_i \ : \ |\cS| \leq \beta\}$.

{Throughout, we consider \textit{completely randomized} experimental designs, $\bz \sim \textrm{CRD}(xn,n)$, wherein a uniform random subset of $xn$ entries of $\bz$ are assigned $1$. We use the notation $\big[ \tfrac{xn}{n} \big]_m = \prod_{i=0}^{m-1} \tfrac{xn-i}{n-i}$ to denote the probability that a subset of $m$ individuals is fully treated under such a design.}
\section{Two-Stage Rollout Designs} \label{sec:newmethod}


\begin{figure}
    \centering
    \includegraphics[width=0.45\textwidth]{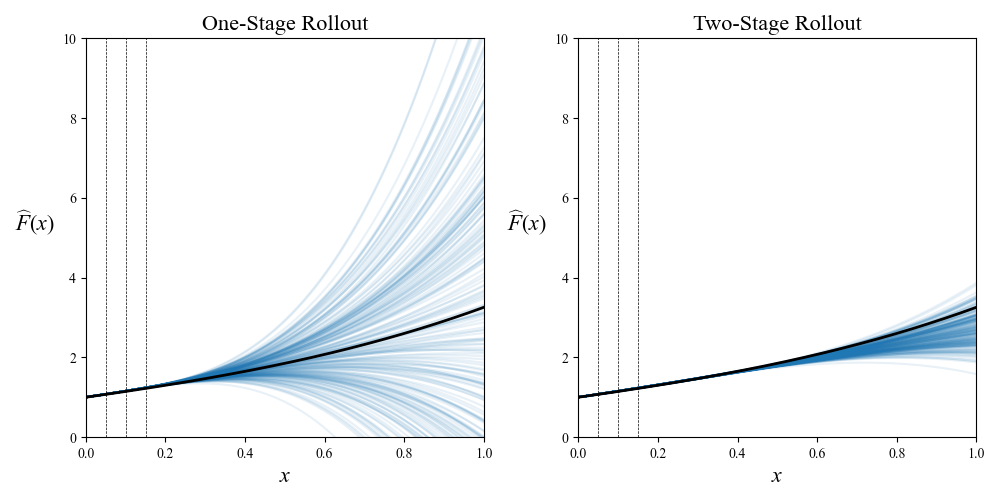}
    \caption{Visualization of extrapolated polynomials used to estimate $\TTE$ across 200 runs of a rollout experiment on a $20\times20$ lattice with $\beta=3$. The left plot uses a one-stage rollout ($p=0.15$), as in~\cite{cortez2022neurips}, while the right plot uses a two-stage rollout ($q=0.375$). The two-stage design incurs bias, but extrapolation in the one-stage design leads to higher variance.}
    \label{fig:visualize_interpolation}
\end{figure}

Under the $\beta$-order interactions model described in Section~\ref{sec:preliminaries}, the quantity
\[
    \E[ \tfrac{1}{n} \textstyle\sum_{i=1}^{n} Y_i(\bz)] = \tfrac{1}{n} \textstyle\sum_{i\in[n]} \textstyle\sum_{\cS \in \cS_i^\beta} c_{i,\cS} \cdot \big[\tfrac{xn}{n}\big]_{|\cS|} =: F(x),
\]
with the expectation taken over $\bz \sim \textrm{CRD}(xn,n)$ is a polynomial $F$ \mcreplace{in the proportion $x$ of individuals receiving treatment and has}{with} degree at most $\beta$. \mcedit{Here, $x$ is a variable that represents the proportion of individuals assigned to treatment.} We can recover this polynomial (along with $\TTE = F(1) - F(0)$) by interpolating through $\beta+1$ evaluation points. \citet{cortez2022neurips} use this observation to describe a TTE estimator that leverages a  \textit{staggered rollout} experimental design, wherein $\beta+1$ observations are taken as treatment is rolled out to increasing proportions of the population. 
This estimator, equation 4 of \citet{cortez2022neurips}, is unbiased with variance $O\big( \frac{d^2 \beta^{2\beta+2}}{n p^{2\beta}}\big)$, where $p$ represents the \textit{treatment budget}, i.e. proportion of the population assigned to treatment in the final stage of the experiment. In each time step $t \in \{0, \hdots, \beta\}$, where $tpn/\beta$ individuals are assigned to treatment, the estimator accrues sampling variance by using the observed mean outcome $\widehat{F}(tp/\beta)$ as a proxy for $F(tp/\beta)$. Then, the extrapolation of these observations to estimate $\widehat{F}(1)$ magnifies this variance by a factor of $(\beta/p)^{2\beta}$ (see Figure~\ref{fig:visualize_interpolation}). To lessen the effects of extrapolation, we would like to sample points from $F(x)$ closer to $1$. However, to adhere to our treatment budget $p$, we must restrict the rollout to a subset of the population. This motivates the following two-stage rollout design.



\begin{definition}[Two-Stage Rollout Design] \label{def:2stage_design} \
    Given a population $[n]$, model degree $\beta$, treatment budget $p$, and parameter $q \in [p,1]$, we consider experimental designs with the following two stages. 

    \textbf{Stage 1:} Select a subset of the population $\cU \subseteq [n]$ with {$|\cU| = pn/q$ and} marginals $\Pr(i \in \mathcal{U}) = p/q$.

    \textbf{Stage 2:} Run a $(\beta+1)$-stage staggered CRD rollout experiment on the units in $\mathcal{U}$, leaving all other units untreated. Such an experiment satisfies:
    \begin{description}
        \item[Treatment Restriction:] $z_i^t = 0$ \hfill $\forall \; t, i \not\in \mathcal{U}$. 
        \item[Per-Round Treatment:] $\bz_{\cU}^t \sim \emph{CRD}\big(\tfrac{tpn}{\beta},|\cU|\big)$ \hfill $\forall \; t$.
        \item[Monotonicity:]  $z_i^t \geq z_i^{t-1}$ \hfill $\forall \; i,t \geq 1.$
    \end{description}
\end{definition}

We interpret the parameter $q$ as the \textit{effective treatment budget} of the units within set $\cU$. \mcedit{The treatment restriction condition emphasizes that only indivudals chosen in the first stage are eligible for treatment in the second stage. The per-round treatment condition ensures that at each time step $t$ of the rollout, $tpn/\beta$ units are selected with a completely randomized design from the set of individuals chosen in the first stage ($\cU$).} When $q=p$, $\mathcal{U}=[n]$ and this design reduces to the CRD (one-stage) staggered rollout designs from \citep{cortez2022neurips}. \mcedit{The monotonicity assumption is natural in settings where once a unit is treated, they are always treated because treatment cannot be ``taken back.''} We consider two variants of the two-stage rollout design.

\begin{example}[Unit CRD Rollout Design] \label{ex:unit2stage}

    Select $\cU$ according to a $\emph{CRD}\big(np/q,n\big)$ design. To realize this design,  one can sample $U_i \sim \emph{Unif}(0,1)$ i.i.d. for each $i \in [n]$, let $\mathcal{U}$ comprise the $np/q$ individuals with highest $U_i$, and let each $\bz^t$ indicate the $tnp/\beta$ individuals with highest $U_i$.
\end{example}

\begin{example}[Clustered CRD Rollout Design] \label{ex:cluster2stage} \

    Partition the individuals into $n_c$ equal-sized clusters. In Stage 1, use a $\emph{CRD}\big(n_c p/q,n_c\big)$ design to select a subset of clusters, and include all individuals from these clusters in $\cU$.
\end{example}

We assume throughout that the parameters are appropriately chosen to make all treatment sizes \mcreplace{integral}{whole numbers}. Now, following the presentation of \citet{cortez2022neurips}, we develop a $\TTE$ estimator for data collected under such an experiment that leverages the connection to Lagrange polynomial interpolation.  
\begin{equation} \label{eq:estimatorDefn}
    \widehat{\TTE} := 
        \tfrac{q}{np} \sum\limits_{t=0}^{\beta} h_{t,q} \sum\limits_{i=1}^{n} Y_i(\mathbf{z}^t),
\end{equation}
{where $h_{t,q} = \prod_{\substack{s=0 \\ s\not=t}}^{\beta} \tfrac{\beta/q-s}{t-s} - \prod_{\substack{s=0 \\ s\not=t}}^{\beta} \tfrac{-s}{t-s}$} \mcedit{come from the Lagrange coefficients}. When $q=p$, this estimator coincides with the polynomial interpolation estimator of~\citet{cortez2022neurips}. The estimator is equivalent to applying the polynomial interpolation estimator with the Stage 2 budget $q$ and then scaling the result by $q/p$, since only a $p/q$ fraction of units are selected in Stage 1 to be eligible for treatment.

The estimate can be evaluated in $O(n\beta)$ time and requires no information about the edges in the interference network. 
While $\beta$ is a parameter of the potential outcomes model, it also appears in the estimator due to its use of polynomial interpolation: fitting a $\beta$-degree polynomial requires $\beta+1$ points. \mcedit{Note that this estimator requires knowledge of $\beta$, as does the estimator in ~\citet{cortez2022neurips}. Determining or choosing $\beta$, while an interesting and practical research direction, is beyond the scope of this paper.}
\mcedit{
\begin{remark}
    The estimator defined in \eqref{eq:estimatorDefn} does not require knowledge of the interference network. However, the design described in Definition \ref{def:2stage_design} may or may not require knowledge of the network depending on how the subset is selected in Stage 1. For example, the design described in Example \ref{ex:unit2stage} does not require knowledge of the interference network, but the design described in Example \ref{ex:cluster2stage} will require graph knowledge if the clustering method requires it.
\end{remark}
}

To analyze a Clustered CRD Rollout Design, we introduce the following notation. A clustering $\Pi$ of the interference network is a partition of $[n]$ into $n_c$ disjoint sets; $\pi \in \Pi$ is a subset $\pi \subseteq [n]$ of units assigned to the same cluster. Given a unit $i \in [n]$, we define $\pi(i)$ as the cluster containing $i$. Given $\cS\!\subseteq\![n]$, we define $\Pi(\cS) := \{ \pi \in \Pi \ | \ \exists \; i\!\in\!\cS \colon \pi = \pi(i) \}$. We define the average treatment effect of a particular cluster $\pi\in\Pi$ as
\[
\bar{L}_\pi := \frac{n_c}{n} \sum\limits_{i\in [n]} \sum\limits_{\cS \in \cS_i^\beta \setminus \varnothing} c_{i,\cS} \cdot \Ind \big(\cS \subseteq \pi \big),
\]
which is equivalent to the portion of the $\TTE$ contained in a single cluster $\pi$.

Throughout our analysis, it will be useful to consider the \textit{cut effect}, defined as:
\[
    C(\delta(\Pi)) := \tfrac{1}{n} \sum\limits_{i\in [n]} \sum_{\cS \in \cS_i^\beta} c_{i,\cS} \cdot \Ind \big(|\Pi(\cS)| \geq 2 \big).
\]
The cut effect $C(\delta(\Pi))$ denotes the average treatment effect attributable to subsets $\cS$ that span across multiple clusters. Since each $\cS \subseteq \cN_i$ for some $i \in [n]$, $C(\delta(\Pi)) = 0$ when there are no edges cut by the clustering, and it increases with the number of cut edges.
{Note that the cut effect can be equivalently expressed as $C(\delta(\Pi))=\TTE-\tfrac{1}{n_c}\sum_{\pi\in\Pi}\bar{L}_\pi$.}
\section{Theoretical Results} \label{sec:theoretical_results}

In this section, we consider the bias and variance of estimator \eqref{eq:estimatorDefn}, with a focus on the Clustered CRD Rollout Design from Example~\ref{ex:cluster2stage}. \mcedit{In our results, we notationally distinguish between theoretical variance ($\Var$) and empirical variance ($\widehat{\Var}$) with the hat notation.} Our first theorem gives a general expression for the bias of estimator \eqref{eq:estimatorDefn} under two-stage rollout designs.

\begin{theorem} \label{thm:bias}
    Under a $\beta$-order potential outcomes model and a Two-Stage Rollout Design, estimator \eqref{eq:estimatorDefn} has bias
    \[
        \tfrac{1}{n} \sum_{i\in[n]} \sum_{\substack{\cS \in \cS_i^\beta \\ \cS \ne \varnothing}} c_{i,\cS} \Big[\tfrac{q}{p} \cdot \Pr\big( \cS \subseteq \cU \big)-1\Big].
    \]
\end{theorem}

{A proof appears in Appendix~\ref{app:proofs} and uses the law of total expectation to derive an expression for $\E\big[\widehat{\TTE}\big]$ by first conditioning on $\cU$. 
When $|\cS| = 1$, i.e. $\cS = \{j\}$, the marginal condition $\Pr\big(j \in \cU \big) = p/q$ in Stage 1 of the experiment ensures that these terms will not contribute any bias. Rather, all of the bias comes from larger subsets $\cS$. Intuitively, bias arises from the possibility that an individual's neighborhood can be partially within and partially outside of $\cU$; the estimation of such an individual $i$'s treatment effect via interpolation will be biased as the rollout proportion in each time step will not match the expected proportion of $\cN_i$ that is treated. However, when $q=p$, $\cU = [n]$ such that $\Pr( \cS \subseteq \cU )=1$ and the estimator is unbiased.}  

The next two corollaries specialize Theorem~\ref{thm:bias} to the two-stage designs described in Examples \ref{ex:unit2stage} and \ref{ex:cluster2stage}.

\begin{corollary}
    Under a $\beta$-order potential outcomes model, the bias of \eqref{eq:estimatorDefn} under a two-stage Unit CRD Rollout Design is
    \[
        \tfrac{1}{n} \sum_{i\in[n]} \sum_{\substack{\cS \in \cS_i^\beta\setminus \varnothing }} c_{i,\cS} \Big[\tfrac{q}{p} \cdot \big[ \tfrac{(p/q)n}{n} \big]_{|\cS|} - 1 \Big].
    \]
\end{corollary}

\begin{corollary}
    Under a $\beta$-order potential outcomes model and a clustering $\Pi$ of the interference network into $n_c$ equal-sized clusters, the bias of \eqref{eq:estimatorDefn} under a two-stage Clustered CRD Rollout Design is
    \[
       \tfrac{1}{n} \sum_{i\in[n]} \sum_{\cS \in \cS_i^\beta \setminus \varnothing} c_{i,\cS} \Big[\tfrac{q}{p} \cdot \big[\tfrac{(p/q)n_c}{n_c}\big]_{|\Pi(\cS)|} - 1\Big].
    \]
\end{corollary}

When each $c_{i,\cS} \geq 0$, the bias is always negative, as we are omitting some of the effects corresponding to sets $\cS$ for which $|\Pi(\cS)| \geq 2$. The magnitude of the bias can be upper-bounded by { $\tfrac{q-p}{q} \cdot C(\delta(\Pi)).$}
Next, we present a general bound on the variance of estimator \eqref{eq:estimatorDefn}.



\begin{theorem} \label{thm:var_bounds}
    Under a $\beta$-order potential outcomes model with each $c_{i\cS} \geq 0$ and
     a clustering $\Pi$ of the interference network into $n_c$ equal-sized clusters, we bound the variance of \eqref{eq:estimatorDefn} under a two-stage {Clustered} CRD Rollout Design by:
    \begin{align*}
        &\Ind(q < 1) \cdot \tfrac{ q^3 \beta^2 Y_{\max}^2}{ p^2 n} \Big( \tfrac{\beta}{q} \Big)^{2\beta} \big( d^2 + 4 \beta^3 \big) 
        \ + \ \tfrac{q-p}{p(n_c -1)} \widehat{\emph{Var}}(\bar{L}_\pi) \\
        &
        \ + \ \Ind(q > p) \cdot \tfrac{2d^2 Y_{\max}}{n_c} \cdot C(\delta(\Pi)).
    \end{align*}
    {where $Y_{\max}$ is a bound on the outcomes}.

\end{theorem}

The first term, which has an exponential dependence on $\beta$, is from upper bounding the $\E_{\cU}[\Var_{\bz}(\hTTE \ | \ \cU)]$ term from the law of total variance. As this term grows exponentially with $\beta$, it is large when $\beta$ is not small. Since it also does not depend on the clustering, we can only make it small by choosing $q$ close to 1. When $q=1$ {all selected units are assigned treatment}, so $\E_{\cU}[\Var_{\bz}(\hTTE \ | \ \cU)] = 0$ because there is no randomness in the second stage conditioned on $\cU$.

The second and third terms are from upper bounding $\Var_{\cU}[\E_{\bz}(\hTTE \ | \ \cU)]$. When $q=p$, $\cU = [n]$ is deterministic, such that $\Var_{\cU}[\E_{\bz}(\hTTE \ | \ \cU)] = 0$. When $q > p$, these terms reflect the impact of the clustering on the performance of the estimator. The second term is small if $\widehat{\Var}(\bar{L}_{\pi})$ is small while $n_c$ is still large. Intuitively $\widehat{\Var}(\bar{L}_{\pi})$ is small if the average effects within clusters is well-balanced across clusters, i.e. clusters are similar to each other with respect to their average treatment effects. If covariates are positively correlated with the treatment effects, this encourages clusters that have good covariate balance. The third term depends on the magnitude of the effects from sets $\cS$ with membership in more than one cluster, which is the same expression that showed up in the upper bound on the magnitude of the bias. This encourages clusters that minimize cut edges. 

When the {interference network exhibits strong homophily}, these two clustering objectives are in tension: minimizing cut edges may decrease the cut effect, but increase $\widehat{\Var}(\bar{L}_{\pi})$ by lowering covariate balance. This suggests that traditional clustering algorithms that focus only on graph-based objectives like minimizing the cut may not be optimal. This is important because the literature on graph clustering often focuses on clustering objectives related to graph structure (e.g. edges). Meanwhile, covariate balance is important in causal inference settings where potential outcomes may be correlated with covariates. At the intersection of these two research areas is causal inference under network interference, where an effective clustering should capture more than just graph structure when there is homophily. This is an important reminder that with cluster-randomized designs for causal inference under network interference, considering \textit{both} graph structure and covariate balance may be crucial.

\mcedit{
\begin{remark}
\end{remark}
When we plug $q=p$ into the bound from Theorem \ref{thm:var_bounds}, we obtain variance bound
$\frac{ \beta^{2\beta+2} Y_{\max}^2}{p^{2\beta-1} n} \big( d^2 + 4 \beta^3 \big)$, which we can compare to the asymptotic variance bound $O \Big( \frac{d^2 \beta^{2\beta+2} Y_{\max}^2}{p^{2\beta} n}  \Big)$ from \citet{cortez2022neurips}. If we assume we are in the $\beta \ll d$ regime, then these bounds differ by a factor of $\frac{1}{p}$.}


\subsection{The $\beta=1$ Setting}

Here, we strengthen our variance bounds for the setting of linear ($\beta = 1$) heterogeneous outcomes models. As noted above, estimator \eqref{eq:estimatorDefn} is unbiased in this setting. \cydelete{To reason about the variance, we will follow the perspective of \citet{yu2022PNAS}.} For each $j \in [n]$, let us introduce the quantity 
$L_j = \textstyle\sum_{i \colon j \in \cN_i} c_{i,\{j\}}$
to represent the total \textit{outgoing} effect that treating $j$ has on the population. 
\cydelete{Note that $\TTE = \tfrac{1}{n} \sum_{j} L_j$, and we may re-express estimator \eqref{eq:estimatorDefn}  as
\[
    \widehat{\TTE} = \tfrac{1}{np} \textstyle\sum_{j} L_j z_j^1.
\]}
We use Lemma \ref{lemma:CRD_PNAS}, restated from \citet{yu2022PNAS} to understand the variance of $\widehat{\TTE}$. 
\begin{lemma} \label{lemma:CRD_PNAS}
    Suppose that $\bz \sim \emph{CRD}(p \cdot |\bz|,|\bz|)$. Then, 
    \[
        \emph{Var}\big( \tfrac{1}{p \cdot |\bz|} \textstyle\sum_{i} L_i z_i\big) = \tfrac{1-p}{p \cdot ( |\bz| - 1 )} \cdot \widehat{\emph{Var}} \big( L_i \big),
    \]
    where $\widehat{\emph{Var}} \big( L_i \big) = \tfrac{1}{|\bz|} \sum_{j} (L_i)^2 - \big( \tfrac{1}{|\bz|} \sum_{j} (L_i \big)^2$ and $|\bz|$ is the total number of entries in $\bz$.
\end{lemma}

We use this lemma to derive the variance expression in the following theorem.
{
\begin{theorem} \label{thm:variance-deg1-crd} 
Under a potential outcomes model with $\beta=1$ and a two-stage Clustered CRD rollout design with clustering $\Pi$, estimator \eqref{eq:estimatorDefn} has variance
\[
    \tfrac{1-q}{pn - q} \cdot \widehat{\underset{j \in [n]}{\emph{Var}}}(L_j) \ + \ \tfrac{(q-p)(pn-1)}{p(n_c-1)(pn-q)} \cdot \widehat{\underset{\pi \in \Pi}{\emph{Var}}}\big(\bar{L}_\pi\big).
\]
\end{theorem}
}
{The proof is in Appendix \ref{app:proofs}.
The \mcedit{$\widehat{\Var}_{j\in[n]}(L_j)$} term is the \mcreplace{population-wide}{empirical} variance of treatment effects \mcedit{across the population} and comes from applying Lemma \ref{lemma:CRD_PNAS} where the outer sum is over units $i\in[n]$. The $\widehat{\Var}(\bar{L}_\pi)$ term is the across-cluster variance of average cluster treatment effects and comes from applying the lemma where the outer sum is indexed over clusters $\pi \in \Pi$.
When $q=1$, the design is a simple cluster randomized design where every selected cluster in the first stage is treated in the second stage, and the variance expression simplifies to $\tfrac{1-p}{p(n_c-1)} \widehat{\Var}_{\pi\in\Pi}\big(\widehat{\E}_{j\in\pi}[L_j]\big)$. This aligns exactly with the cluster-randomized result from \citet{yu2022PNAS}. 
When $q=p$, our variance expression is $\tfrac{1-p}{p(n-1)} \widehat{\Var}_{j\in[n]}(L_j)$ and aligns exactly with the completely randomized design result from \citet{yu2022PNAS}.
When $q\in(p,1)$, some of the variance is attributable to the population-wide variance of influences (which we have no control over) and some of the variance is attributable to variance of average influences across clusters, which can be controlled with a clustering that enforces covariate balance in settings where covariates positively correlate with treatment effects. Note that cut edges do not play a role in the bias or variance when $\beta=1$.
}

    


\section{Experiments} \label{sec:experiments}

In this section, we use experiments on both synthetic and real-world networks to analyze the performance of the two-stage estimator.

\paragraph{Potential Outcomes Model.} We use synthetic potential outcomes that generalize the response model of \citet{ugander2023randomized} to incorporate $\beta$-order interactions; refer to Section 6.2 in their paper for an in-depth description of the design choices of this model. {The model incorporates homophily, degree-correlated effects, and $\beta$-order interference.} Unless otherwise noted, our choices for the parameter values \mcdelete{in this} agree with \citet{ugander2023randomized}. Refer to Appendix \ref{app:networks} for further details about the model and parameters. 

\paragraph{Networks.}

The synthetic networks that we consider are {$\sqrt{n} \times \sqrt{n}$ lattice graphs}. We include all self-loops in these networks.

In addition, we consider three real-world networks: an email communication network \citep{leskovec2007graph}, a social network \citep{network_repo}, and a co-purchase network in an online marketplace \citep{leskovec2007dynamics}; details of these networks can be found in Appendix~\ref{app:networks}. Each dataset includes a network and a set of feature labels $F$ assigned to its vertices. In our experiments, we sometimes use these features to cluster the networks.

\paragraph{Running the Experiments.} The source code for our experiments and plots found within our manuscript is available in the supplementary materials. All of our experiments were run on a MacBook Pro with an M3 chip and 16GB of memory and ran in under two hours parallelized across its 8 cores.

\subsection{Comparison with other estimators}
\label{ssec:experiment-compare-est}

We first empirically explore the performance of our two-stage approach without clustering. We compare the bias and variance of the following estimators: 

\begin{itemize}
    \item \textsf{2-Stage}, the Polynomial Interpolation (PI) estimator under a unit two stage rollout (with no clustering) and $q=0.5$
    \item Two difference-in-means style estimators. The classical \textsf{DM} estimator, and a thresholded version, \textsf{DM}($\lambda$) that only considers individuals for which a $\lambda$-proportion of their neighborhood shares their treatment assignment.
    \item The \textsf{H\'{a}jek} estimator, an inverse probability weighted-style estimator.
    \item The \textsf{PI} estimator under a {one-stage $\textrm{CRD}(pn,n)$ rollout over $\beta+1$ time steps \mcedit{from \citet{cortez2022neurips}}.}
    \item {The two-stage PI estimator with \textsf{q=1}, equivalent to the PI estimator under a one-stage $\textrm{CRD}(pn,n)$ rollout over $2$ time steps (with no clustering).}
\end{itemize}

The exact formulas for these estimators can be found in Appendix \ref{app:other_ests}. Although the Horvitz-Thompson estimator is also considered a baseline, due to its high variance, it consistently performs worse than all the other estimators considered so it is omitted. Of the non-PI estimators, only the (unthresholded) difference in means does not require knowledge of the underlying interference network. \cydelete{Only the Horvitz-Thompson estimator is unbiased.} 

Figure~\ref{fig:amazon} shows the bias and standard deviation of these estimators as we vary the treatment budget $p$. The column faceting distinguishes between the cases of $\beta=1$, $\beta=2$, and $\beta=3$. While the difference in means estimators have low variance, their bias leads to higher mean squared error (MSE) than the two-stage estimator. 
\begin{figure}
    \centering
    \includegraphics[width=0.5\textwidth]{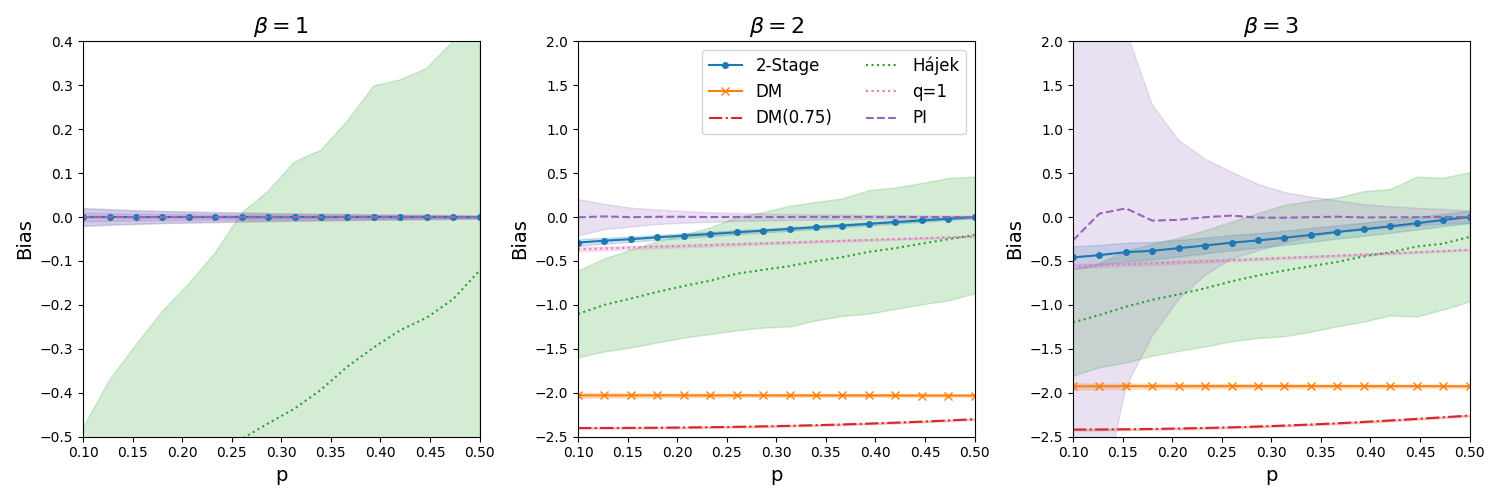} 
    \caption{Performance of different estimators on the \textsc{Amazon} network for various values of $p$. The bold line indicates the mean over 1000 replications. The shading indicates the experimental standard deviation, calculated by taking the square root of the experimental variance over all replications. The \textsf{2-Stage} estimator uses $q=0.5$ and does not utilize clustering. {Note the scaling of the $y$-axes are not the same across $\beta$.}}
    \label{fig:amazon}
\end{figure}
{In the $\beta=1$ case, we have zoomed in on the scaling since the variance of the polynomial interpolation estimators is very small. In this setting they are unbiased. \textsf{PI} is equivalent to \textsf{q=1} in this case so they perfectly overlap and have smaller variance than \textsf{2-Stage}. This suggests that when you have a truly linear model, the two-stage approach may not improve over the one-stage approach. The remaining estimators either have much worse variance or much worse bias, and thus they do not show up on the plot.}

{In the $\beta=2$ case, the difference in means estimators are biased with very low variance. The \textsf{H\'{a}jek} estimator has bias that decreases as $p$ increases, but significantly higher variance than all other approaches. Between the three polynomial interpolation-based estimators, \textsf{PI} is unbiased with slightly larger variance for small values of $p$, while \textsf{2-Stage} has bias that decreases as $p$ increases and \textsf{q=1} has bias that remains about the same regardless of $p$. In this case, the MSE of these estimators is relatively similar, with \textsf{PI} doing slightly better than \textsf{2-Stage} for smaller values of $p$, again suggesting this is not a setting where the two-stage approach \mcedit{necessarily} does better. }

{In the $\beta=3$ case, we can see the significant variance reduction of the two-stage estimator over \textsf{PI}, which comes at the expense of a smaller introduction of bias relative to the remaining estimators. Due to the richer model, we see that the one-stage approach (\textsf{PI}) has an extremely high variance for smaller values of $p$, much larger than the bias incurred by the two-stage approach.} 

Overall, the performance of the \textsf{2-Stage} and \textsf{PI} estimators is similar for most of the parameter landscape {when $\beta=1$ or $\beta=2$}, but the variance reduction of the \textsf{2-Stage} estimator for small $p$ and $\beta = 3$ results in a lower MSE despite the additional bias. {This makes sense because we only expect a large error reduction for richer models and small treatment probabilities.}

These results highlight a setting where the two-stage approach improves over the one-stage approach, even without network information, as these experiments did not use any clustering. Plots of the MSE and the two other network datasets can be found in Appendix~\ref{app:experiments}.


\subsection{Clustering effect in two-stage estimator}

In this section we conduct experiments to empirically explore the impact of clustering. 

\paragraph{Lattice.} In Figure \ref{fig:compare_lattice_clusterings}, we compare the MSE of the 2-Stage estimator under two clusterings and no clustering on a $100\times100$ lattice. The \textsf{Coarse} clustering is a $10\times10$ grid on top of the lattice; there are 100 clusters with 100 people in each cluster. The \textsf{Fine} clustering is a $2\times2$ grid on top of the lattice; there are 2500 clusters with 4 people in each cluster. Table \ref{tab:clustering_metrics_lattice} displays some metrics for these clusterings: $\widehat{\Var}\big(\bar{L}_\pi\big)$, $C(\delta(\Pi))$, and the number of cut edges. 

\begin{table}[h]
        \centering
        \caption{Clustering Metrics for Figure \ref{fig:compare_lattice_clusterings}.}
        \label{tab:clustering_metrics_lattice}
        
        \begin{tabular}{c|c|c|c}
            Clustering & $\widehat{\Var}\big(\bar{L}_\pi\big)$ & $C(\delta(\Pi))$ & Cut Edges \\
            \hline
            \textsf{Coarse} &  0.0002 & 0.1229 & 3600 \\
            \textsf{Fine} & 0.002 & 0.5703 & 19600
        \end{tabular}
\end{table}

\begin{figure}
    \centering
    \includegraphics[width=0.45\textwidth]{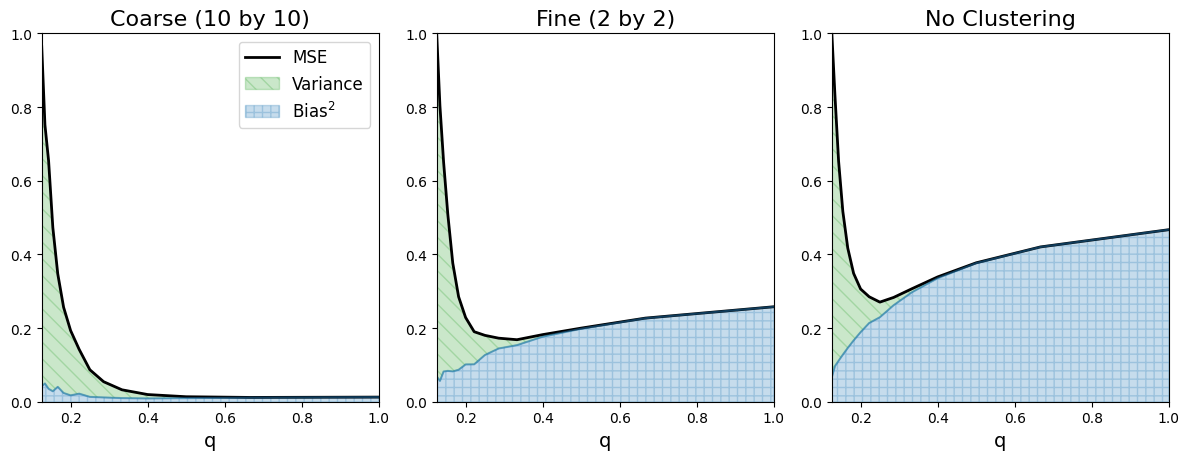} 
    \caption{Mean Squared Error of the Two-Stage TTE estimator for two clusterings of a $100\times100$ lattice graph, compared with no clustering, for a $\beta$-degree potential outcomes model with $\beta=3$. Even with no network knowledge, we see a drastic decrease in MSE even at the cost of incurring bias. 
    }
    \label{fig:compare_lattice_clusterings}
\end{figure}

{
In Figure \ref{fig:compare_lattice_clusterings}, we vary $q$ (the treatment probability for individuals selected in the first stage) on the $y$-axis and plot the MSE, with the different shading corresponding to the variance and squared bias components. The leftmost endpoint corresponds to $q=p=0.15$, equivalently the one-stage setting from \cite{cortez2022neurips}. Since the treatment budget $p$ is small and $\beta=3$ (indicating a richer model), the left side of each plot exhibits high variance and low bias. As $q$ increases, the variance decreases but the bias increases due to cut edges. Clustering reduces bias by reducing the number of cut edges. The coarse clustering in the left plot drastically decreases the error, especially as $q$ approaches $1$. The middle plot is a finer clustering, and results in much more bias as $q$ approaches $1$. 
Table \ref{tab:clustering_metrics_lattice} helps elucidate the difference in performance. The fine clustering cuts five times more edges than the coarse clustering, resulting in a cut effect that is about five times larger. Finally, the rightmost plot shows the MSE of the two-stage estimator under no clustering, i.e. a unit CRD 2-stage rollout, and incurs the largest amount of bias. Overall, the error for $q>p$ is smaller than at $q=p$ across all plots, showing settings where a two-stage design leads to improvement over a one-stage design.
}

\paragraph{Real-world Networks.}


We compare two methods of clustering the real-world networks. In the clustering with \textsf{Full Graph Knowledge}, we cluster the true underlying graph using the \textsc{METIS} clustering library by \citet{karypis1998fast}. In the clustering with \textsf{Covariate Knowledge}, clusters are based on features. When each vertex is assigned to one feature, we use these assignments as the clustering. When vertices may have multiple features we form a feature graph --- a weighted graph, where the weight of edge $(i,j)$ is the number of feature labels shared by $i$ and $j$ --- and cluster this feature graph using \textsc{METIS}.

\begin{figure}
    \centering
    \includegraphics[width=0.45\textwidth]{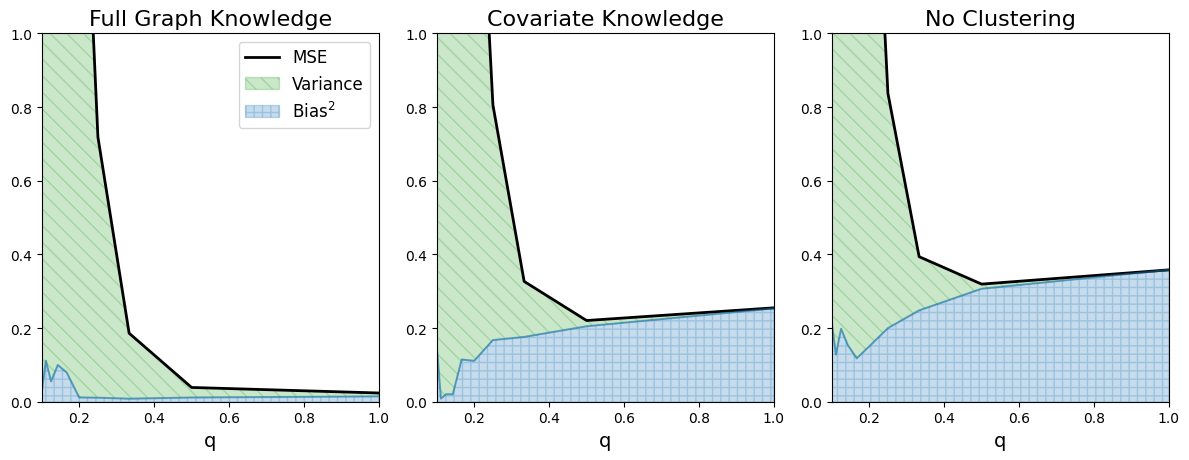}
    \caption{Mean Squared Error of the Two-Stage TTE estimator for two clusterings (with 250 clusters) of the \textsc{Amazon} network, compared with no clustering, for a $\beta$-degree potential outcomes model with $\beta=3$. Even with no network knowledge, we see a drastic decrease in MSE even at the cost of incurring bias.
    }
    \label{fig:compare_amazon_clusterings}
\end{figure}

We highlight the Amazon network here, but additional experiments with the other networks are in Appendix \ref{app:experiments}. Figure~\ref{fig:compare_amazon_clusterings} 
depicts the results of such an experiment run on a co-purchasing network of Amazon products.  We compare the two clusterings of this graph as described above, each time partitioning the network into 250 parts, against no clustering.
To generate these plots, we compute the experimental bias, sampling variance, and total variance,  over $1000$ replications. 
Here, we vary $q$ from $q=p=0.1$ to $q=1.$
In the first plot, showing the MSE of the estimator when the clustering uses full network knowledge, the MSE is minimized at $q=1$ with value 0.024. In the second plot, showing the MSE of the estimator when the clustering only uses covariate knowledge, the MSE is minimized around $q=0.5$ with value 0.22. 
{Table \ref{tab:clustering_metrics_amazon} gives insight into the difference in performance under these clusterings. The clustering with covariate knowledge cuts about five times as many edges as the clustering with full graph knowledge, resulting in a cut effect that is about 5 times larger.}

\begin{table}[h]
        \centering
        \caption{Clustering Metrics for Figure \ref{fig:compare_amazon_clusterings}.}
        \label{tab:clustering_metrics_amazon}
         
        \begin{tabular}{c|c|c|c}
            Cluster & $\widehat{\Var}\big(\bar{L}_\pi\big)$ & $C(\delta(\Pi))$ & Cuts \\
            \hline
            \textsf{Full} & 0.2488 & 0.1258 & 7670\\
            \textsf{Covariate} & 0.0426 & 0.5436 & 41243\\
        \end{tabular}
\end{table}

In the final plot,  showing the MSE of the estimator under a two-stage unit CRD design, the MSE is minimized around $q=0.5$ with value 0.32. Recall that the leftmost endpoint of each plot corresponds to the error when $q=p$, i.e. under the one-stage rollout. Although a clustering with full network knowledge achieves the best overall performance, we see a significant error reduction over a one-stage even for a two-stage unit CRD design. Thus, using the two-stage estimator may reduce MSE (versus a single-stage rollout) even without a clustering or network knowledge. 

\begin{table}[h]
    \centering
    \caption{Clustering Metrics for Amazon Network}
    \label{tab:clustering_metrics_amazon_nc}
    
    \begin{tabular}{c|c|c|c|c|c}
        Cluster & $n_c$ & $\widehat{\Var}\big(\bar{L}_\pi\big)$ & $C(\delta(\Pi))$ & $q_{\min}$ & MSE \\
        \hline
        \textsf{Full} & 50 & 0.156059 & 0.088226 & 1 & 0.035\\
        \textsf{Full} & 100 & 0.187048 & 0.102260 & 1 & 0.028\\
        \textsf{Full} & 250 & 0.248759 & 0.125772 & 1 & 0.024\\
        \textsf{Covariate} & 50 & 0.019543 & 0.517855 & 0.5 & 0.211\\
        \textsf{Covariate} & 100 & 0.025207 & 0.536876 & 0.5 & 0.228\\
        \textsf{Covariate} & 250 & 0.042644 & 0.543623 & 0.5 & 0.220
    \end{tabular}
\end{table}


{In Table \ref{tab:clustering_metrics_amazon_nc}, we record some metrics of clusterings of different sizes computed with full network or covariate knowledge.
These experiments use a $\beta$-degree potential outcomes model with $\beta=3$. The parameter $n_c$ indicates the number of clusters. In each row, $q_{\min}$ is the value of $q$ that minimizes the MSE and the column MSE contains that value. We computed the minimum empirically. Generally, having more clusters corresponds to a higher across-cluster variance of average cluster influences and a higher cut effect. However, regardless of cluster size, the MSE is still drastically decreased from 38 (at $q=p)$ to about 0.2 under a covariate-based clustering (at $q=0.5$) and to about 0.02 with a full graph knowledge-based clustering (at $q=1$). Clustering with full knowledge has higher across-cluster variance of average cluster influences but smaller cut effect compared with clustering with covariate knowledge. This reminds us that there is a tension between cut edges and covariate balance. While a common clustering objective is to minimize cut edges, there may be settings where enforcing some covariate balance may be wise if there is homophily. This is because if there is strong homophily, edges are correlated with covariates. If there is reason to believe these covariates are highly correlated with potential outcomes, then minimizing cut edges might minimize the cut effect but maximize the variance of cluster influences.}

\mcedit{Our theoretical and experimental results explore settings with equal-sized clusters, such as in \cite{eckles2017design,brennan2022cluster,candogan2024correlated}, which may be too difficult a constraint to meet in many practical settings. In theory, unequal-size clusters should not affect bias but will affect the variance of the estimates. For the experiments on real-world networks, we use the METIS clustering library, which can only do equal-size clusters. Exploring the performance under unequal-sized clusters is a practical direction for future work. 
}


\subsubsection*{Acknowledgements}
 We gratefully acknowledge financial support from the National Science Foundation grants CCF-2337796 and CNS-1955997, the National Science Foundation Graduate Research Fellowship grant DGE-1650441, and AFOSR grant FA9550-23-1-0301.

\bibliography{refs}

\section*{Checklist}

 \begin{enumerate}
 
 \item For all models and algorithms presented, check if you include:
 \begin{enumerate}
   \item A clear description of the mathematical setting, assumptions, algorithm, and/or model. [Yes, see Section \ref{sec:preliminaries}.]
   \item An analysis of the properties and complexity (time, space, sample size) of any algorithm. [Not Applicable]
   \item (Optional) Anonymized source code, with specification of all dependencies, including external libraries. [Yes, in supplementary material.]
 \end{enumerate}

 \item For any theoretical claim, check if you include:
 \begin{enumerate}
   \item Statements of the full set of assumptions of all theoretical results. [Yes]
   \item Complete proofs of all theoretical results. [Yes, in Appendix]
   \item Clear explanations of any assumptions. [Yes]     
 \end{enumerate}

 \item For all figures and tables that present empirical results, check if you include:
 \begin{enumerate}
   \item The code, data, and instructions needed to reproduce the main experimental results (either in the supplemental material or as a URL). [Yes, in supplementary material and as URL in camera-ready version.]
   \item All the training details (e.g., data splits, hyperparameters, how they were chosen). [Yes, in Appendix and in Section \ref{sec:experiments}]
         \item A clear definition of the specific measure or statistics and error bars (e.g., with respect to the random seed after running experiments multiple times). [Yes]
         \item A description of the computing infrastructure used. (e.g., type of GPUs, internal cluster, or cloud provider). [Yes]
 \end{enumerate}

 \item If you are using existing assets (e.g., code, data, models) or curating/releasing new assets, check if you include:
 \begin{enumerate}
   \item Citations of the creator If your work uses existing assets. [Yes]
   \item The license information of the assets, if applicable. [Not Applicable]
   \item New assets either in the supplemental material or as a URL, if applicable. [Yes]
   \item Information about consent from data providers/curators. [Yes]
   \item Discussion of sensible content if applicable, e.g., personally identifiable information or offensive content. [Not Applicable]
 \end{enumerate}

 \item If you used crowdsourcing or conducted research with human subjects, check if you include:
 \begin{enumerate}
   \item The full text of instructions given to participants and screenshots. [Not Applicable]
   \item Descriptions of potential participant risks, with links to Institutional Review Board (IRB) approvals if applicable. [Not Applicable]
   \item The estimated hourly wage paid to participants and the total amount spent on participant compensation. [Not Applicable]
 \end{enumerate}

 \end{enumerate}

\newpage
\onecolumn
\appendix
\aistatstitle{Supplementary Materials}

\section{PROOFS} \label{app:proofs}

\subsection*{Theorem~\ref{thm:bias}}
\label{ssec:bias-proof}
\begin{proof} 
    We use the Law of Total Expectation, conditioning on the set of individuals $\cU$ selected in the first stage and reasoning about the randomness from the treatment assignments $\bz$. We have,
    \begin{align}
        \E_{\bz} \big[ \hTTE \:\big|\: \cU \big] 
        &= \tfrac{q}{np} \sum\limits_{t=0}^{\beta} h_{t,q} \sum\limits_{i=1}^{n} \sum_{\cS \in \cS_i^\beta} c_{i,\cS} \cdot \E_{\bz} \Big[ \prod_{j \in \cS} z_j^t \:\big|\: \cU \Big] \nn \\
        &= \tfrac{q}{np} \sum\limits_{t=0}^{\beta} h_{t,q} \sum\limits_{i=1}^{n} \sum_{\cS \in \cS_i^\beta} c_{i,\cS} \cdot \Ind(\cS \subseteq \cU) \cdot \Big[ \tfrac{tpn/\beta}{|\cU|} \Big]_{|\cS|} \nn \\
        &= \tfrac{q}{np} \sum\limits_{i=1}^{n} \sum_{\cS \in \cS_i^\beta} c_{i,\cS} \cdot \Ind(\cS \subseteq \cU) \sum\limits_{t=0}^{\beta} h_{t,q} \Big[ \tfrac{tpn/\beta}{|\cU|} \Big]_{|\cS|} \nn \\
        &= \tfrac{q}{np} \sum\limits_{i=1}^{n} \sum_{\cS \in \cS_i^\beta} c_{i,\cS} \cdot \Ind \big( \cS \subseteq \cU \big) \Big( 1^{|\cS|} - 0^{|\cS|} \Big) \nn \\
        &= \tfrac{q}{np} \sum\limits_{i=1}^{n} \sum_{\cS \in \cS_i^\beta \setminus \varnothing} c_{i,\cS} \cdot \Ind \big( \cS \subseteq \cU \big) \label{eq:e_tte_cond_u}
    \end{align}
    \meedit{Here, the fourth line follows from the properties of Lagrange interpolation. Note that $h_{t,q} = \ell_{t,\bx}(1) - \ell_{t,\bx}(0)$, where $\ell_{t,\bx}$ is the $t$'th Lagrange basis polynomial with evaluation points $\bx = \big( \tfrac{tq}{\beta} \big)_{t \in 0, \hdots, \beta} = \big( \tfrac{tpn/\beta}{|\cU|} \big)_{t \in 0, \hdots, \beta}$. Thus, for any polynomial $f(x)$ with degree at most $\beta$,
    \[
        \sum_{t=0}^{\beta} h_{t,q} \cdot f\big(\tfrac{tq}{\beta}\big) = f(1) - f(0).
    \]
    In this case, we let $f(x) = \Big[\frac{x|\cU|}{|\cU|}\Big]_{|\cS|}$, to find that 
    \[
        \sum\limits_{t=0}^{\beta} h_{t,q} \Big[ \tfrac{tpn/\beta}{|\cU|} \Big]_{|\cS|} = [1]_{|\cS|} - [0]_{|\cS|} = 1^{|\cS|} - 0^{|\cS|}. 
    \]}
    Now, taking the expectation over the randomness in $\cU$, we obtain
    \begin{align*}
        \E \big[ \hTTE \big] &= \E_{\cU} \Big[ \E_{\bz} \big[ \hTTE \:\big|\: \cU \big] \Big] \\
        &= \tfrac{1}{n} \sum\limits_{i=1}^{n} \sum_{\cS \in \cS_i^\beta \setminus \varnothing} c_{i,\cS} \cdot \tfrac{q}{p} \cdot \Pr(\cS \subseteq \cU).
    \end{align*}
    The bias expression in the theorem statement follows from the expression for $\TTE$ given in \eqref{eq:tte_coeff}.
\end{proof}

\subsection*{Theorem~\ref{thm:var_bounds}}

We also use the following algebraic lemma.

\begin{lemma} \label{lem:crd_interp_bound}
    For all $0 < k \leq \ncU \leq 1$, 
    \[
        |h_{t,q}| = \prod_{\substack{s=0 \\ s\not=t}}^{\beta} \tfrac{\beta/q-s}{t-s} - \prod_{\substack{s=0 \\ s\not=t}}^{\beta} \tfrac{-s}{t-s}
        \:\leq\: \big( \tfrac{\beta}{q} \big)^\beta.
    \]
\end{lemma}

\begin{proof}
    When $\beta = 1$, 
    $|h_{0,q}| = |h_{1,q}| = \tfrac{1}{q}$, so the inequality holds (with equality). Thus, we can restrict our attention to $\beta \geq 2$, for which we consider in two cases. First, if $t \geq 1$, we have
    \begin{equation*}
        \big| h_{t,q} \big|
        = \Big| \prod_{\substack{s = 0 \\ s \not= t}}^{\beta} \tfrac{\beta/q-s}{t-s} \Big|
        \leq \big(\tfrac{\beta}{q}\big)^\beta.
    \end{equation*}
    The equality uses the definition of $h_{t,q}$, and the inequality upper bounds the numerator of each factor with $\beta/q$ and lower bounds the denominator of each factor by $1$. When $t=0$, we apply the triangle inequality to conclude that
    \begin{equation*}
        \big| h_{0,q} \big|
        = \Big| \prod_{s = 1}^{\beta} \tfrac{\beta/q-s}{-s} - 1 \Big|
        \leq \prod_{s=1}^{\beta} \tfrac{\beta}{sq} + 1
        = \tfrac{1}{\beta!} \Big(\tfrac{\beta}{q}\Big)^\beta + 1.
    \end{equation*}
    Since $\beta \geq 2$ and $q \leq 1$, we must have $1 \leq \tfrac{1}{2} \big( \tfrac{\beta}{q} \big)^\beta$. Thus we can upper-bound this last expression by
    \[
        \tfrac{1}{\beta!} \Big(\tfrac{\beta}{q}\Big)^\beta + \tfrac{1}{2} \big( \tfrac{\beta}{q} \big)^\beta = \Big( \tfrac{1}{\beta!} + \tfrac{1}{2} \Big) \cdot \Big(\tfrac{\beta}{q}\Big)^\beta \leq \Big(\tfrac{\beta}{q}\Big)^\beta.
    \]
\end{proof}

We also prove a slightly stronger version of Theorem 3 from \cite{cortez2022neurips} with the constants specified. It first relies on a slightly modified version of Lemma 8 from \cite{cortez2022neurips}. 
\begin{lemma}\label{lem:crazy_bound}
    For any $x \in (0,1]$ and any constants $a,b \in \mathbb{N}$ such that $xn \geq \sqrt{2} ab + b - 1$, 
    \[
       \Bigg| \frac{\Big[ \tfrac{xn-a}{n-a} \Big]_b}{\Big[ \tfrac{xn}{n} \Big]_b} - 1 \Bigg| \leq \tfrac{2 ab}{xn-b + 1},
    \]
\end{lemma}

\begin{proof}
First, let us note that when $a=0$ or $b=0$, both sides of this inequality simplify to $0$, so it holds with equality. Thus, we assume throughout the rest of the proof that $a,b > 0$. Note that our assumption $xn \geq \sqrt{2}ab+b-1$ with $x \leq 1$ implies that $n \geq a+b-1$.

Now, given any any $i \in \{0, \hdots, b-1\}$,
\[
    \frac{xn-a-i}{n-a-i} \leq \frac{xn-i}{n-i} \hspace{10pt} \Rightarrow \hspace{10pt} \frac{\Big[ \tfrac{xn-a}{n-a} \Big]_b}{\Big[ \tfrac{xn}{n} \Big]_b} \leq 1.
\]
As a result, expanding the bracket notation, we have, 
    \begin{align*}
        \Bigg| \frac{\Big[ \tfrac{xn-a}{n-a} \Big]_b}{\Big[ \tfrac{xn}{n} \Big]_b} - 1 \Bigg| 
        &= 1 - \prod_{i=0}^{b-1} \Big( \frac{xn-a-i}{xn-i} \Big) \Big( \frac{n-i}{n-a-i} \Big) \\
        &= 1 -  \prod_{i=0}^{b-1} \Big( 1 - \frac{a}{xn-i} \Big) \underbrace{\Big( 1 + \frac{a}{n-a-i} \Big)}_{\geq 1}\\
        &\leq 1 -  \prod_{i=0}^{b-1} \Big( 1 - \frac{a}{xn-b+1} \Big) \tag{$i \leq b-1$} \\
        &= - \sum_{j=1}^{b} \binom{b}{j} \Big(- \frac{a}{(xn-b+1)} \Big)^j \tag{binomial expansion} \\
        &\leq \sum_{j=1}^{b} \binom{b}{j} \Big(\frac{a}{(xn-b+1)}\Big)^j \cdot \Ind(j \text{ is odd}) \\
        &\leq \Big(\frac{ a b}{xn-b+1} \Big) \sum_{j=0}^{\lfloor (b-1)/2 \rfloor}\Big(\frac{ a b}{xn-b+1} \Big)^{2j} \\
        &\leq \Big(\frac{ a b}{xn-b+1} \Big) \sum_{j=0}^{\lfloor (b-1)/2 \rfloor}\Big(\frac{1}{\sqrt{2}} \Big)^{2j} \tag{$xn \geq \sqrt{2} ab + b - 1$}\\
        &\leq \frac{2 a b}{xn-b+1}. \tag{geometric series with factor $\frac{1}{2}$}
    \end{align*}
\end{proof}

We use this to lemma to give an upper bound on the covariance of two sets being treated under a CRD rollout design with $\tfrac{ptn}{\beta}$ individuals treated in round $t$ for each $t \in \{0,\hdots,\beta\}$.
\begin{lemma}
If $\frac{pt'n}{\beta} \geq 2\beta^2 + \beta - 1$, then for $t \leq t'$ and $\cS \cap \cS' = \emptyset$ with $|\cS|, |\cS'| \geq 1$, it follows that 
\[
    \bigg|\emph{Cov} \Big[ \prod_{j \in \cS} z^t_j , \prod_{j' \in \cS'} z^{t'}_{j'} \Big] \bigg| \leq \frac{4p\beta^3}{n}.
\]
\end{lemma}

\begin{proof}
First, let us note that if $t = 0$, then the first argument of this covariance is not random, so the covariance simplifies to $0$, trivially satisfying the bound. Thus, we may assume that $1 \leq t \leq t'$. We can rewrite the covariance expression:
\begin{align*}
    \bigg| \Cov \Big[ \prod_{j \in \cS} z^t_j , \prod_{j' \in \cS'} z^{t'}_{j'} \Big] \bigg| 
    &= \bigg| \E \Big[ \prod_{j \in \cS} z^t_j \prod_{j' \in \cS'} z^{t'}_{j'} \Big] - \E \Big[ \prod_{j \in \cS} z^t_j \Big] \E \Big[ \prod_{j' \in \cS'} z^{t'}_{j'} \Big] \bigg| \\
    &= \Big[ \tfrac{ptn/\beta}{n} \Big]_{|\cS|} \Big[ \tfrac{pt'n/\beta}{n} \Big]_{|\cS'|} \cdot \left| \frac{\Big[ \tfrac{pt'n/\beta - |\cS|}{n-|\cS|} \Big]_{|\cS'|}}{\Big[ \tfrac{pt'n/\beta}{n} \Big]_{|\cS'|}} - 1\right|.
\end{align*}

We can bound this last absolute value expression using Lemma~\ref{lem:crazy_bound}, letting $x = pt'/\beta$, $a = |\cS|$, and $b = |\cS'|$. Note that $a,b \leq \beta$, so our assumption that $\tfrac{pt'n}{\beta} \geq 2\beta^2 + \beta - 1$ ensures that $xn \geq \sqrt{2}ab + b - 1$. We find that 
\[
    \bigg| \Cov \Big[ \prod_{j \in \cS} z^t_j , \prod_{j' \in \cS'} z^{t'}_{j'} \Big] \bigg| 
    \leq 
    \Big( \tfrac{pt}{\beta} \Big)^{|\cS|} \Big( \tfrac{pt'}{\beta} \Big)^{|\cS'|} \cdot \frac{2|\cS||\cS'|}{\tfrac{pt'n}{\beta} - |\cS'| + 1}
    \leq 
    \frac{2p^2 \beta^3}{pn - \beta^2}
    \leq 
    \frac{4p\beta^3}{n}
\]
Here, the final equality uses the fact that $pn \geq pt'n/\beta \geq 2\beta^2$ to conclude that $\frac{p}{pn-\beta^2} \leq \frac{2}{n}$.
\end{proof}

When $\cS \cap \cS' \neq \emptyset$ for $|\cS|, |\cS'| \geq 1$, it follows that 
\[
    \bigg|\Cov \Big[ \prod_{j \in \cS} z^t_j , \prod_{j' \in \cS'} z^{t'}_{j'} \Big] \bigg| \leq p.
\]
Plugging this into Lemma 6 of \cite{cortez2022neurips}, (so, in their notation, $\alpha = \Big( \tfrac{\beta}{p} \Big)^\beta$, $B_1 = p$, and $B_2 = \frac{4p\beta^3}{n}$), we can upper bound the variance of the staggered rollout estimator under a CRD rollout design by
\begin{equation} \label{eq:refined_rollour_var}
    \Var\big(\widehat{\TTE}\big) \leq \frac{\beta^2 Y_{\max}^2 p}{n} \cdot \Big( \tfrac{\beta}{p} \Big)^{2\beta} \cdot \Big( d^2 + 4\beta^3 \Big).
\end{equation}

\mcdelete{\textbf{Please Note:} There is a small typo in the variance bound given in the statement of this theorem in the paper body. A correct bound, which is proven below is 
\mcreplace{\[
   \Var\Big(\widehat{\TTE}\Big) 
   \leq \Ind(q < 1) \cdot \tfrac{ q^3 \beta^2 Y_{\max}^2}{ p^2 n} \Big( \tfrac{\beta}{q} \Big)^{2\beta} \big( d^2 + 4 \beta^3 \big)
   + \tfrac{q-p}{p(n_c -1)} \widehat{\Var}(\bar{L}_\pi)
   + \Ind(q > p) \cdot \tfrac{2d^2 Y_{\max}}{n_c} \cdot C(\delta(\Pi)). 
\]}
{
\[
   \Var\Big(\widehat{\TTE}\Big) 
   \leq \Ind(q < 1) \cdot \tfrac{ q^3 \beta^2 Y_{\max}^2}{ p^2 n} \Big( \tfrac{\beta}{q} \Big)^{2\beta} \big( d^2 + 4 \beta^3 \big)
   + \tfrac{q-p}{p(n_c -1)} \widehat{\Var}(\bar{L}_\pi)
   + \Ind(q > p) \cdot \Big(\tfrac{\beta d}{n_c} + \tfrac{d^2}{n_c}\Big) \cdot Y_{\max} \cdot C(\delta(\Pi)). 
\]
}

We will correct the statement in the final version of the paper. Importantly, the interpretation of this variance bound remains the same.}

\begin{proof}[Proof of Theorem~\ref{thm:var_bounds}] \phantom{a} \\

    By the Law of Total Variance, we have
    \[
        \Var_\bz\big(\widehat{\TTE}\big) = \E_\cU \Big[ \Var_{\bz | \cU} \Big( \widehat{\TTE} \Big) \Big] + \Var_\cU \Big( \E_\bz \big[ \widehat{\TTE} \; \big| \; \cU \big] \Big).
    \]
    We separately bound each of these terms. 

    \underline{First Term:}

    \meedit{First, let us note that when $q=1$, $h_{0,q} = -1$, $h_{\beta,q} = 1$ and $h_{t,q} = 0$ for all $0 < t < \beta$. In this case, we may simplify the estimator to 
    \[
        \widehat{\TTE} = \tfrac{1}{np} \sum_{i=1}^{n} Y_i(\bz^\beta) - Y_i(\bz^0).
    \]
    Conditioned on $\cU$, this quantity is deterministic, since $z_j^\beta = \Ind(j \in \cU)$ and $z_j^0 = 0$. Thus, the variance of the estimator conditioned on $\cU$ is $0$, making the first term of our variance expression $0$. Thus, we may restrict our attention to the case when $q < 1$ and multiply the resulting expression by the indicator $\Ind(q < 1)$ in our final bound.
    }
    
    Now, let $\tilde{\mathbf{z}} \sim \textrm{CRD}(qn,n)$ be a random vector with $z_j^t = \tilde{z}_j^t \cdot \Ind(j \in \cU).$ Conditioned on $\cU$, we may rewrite our estimator:
    \begin{align*} 
        \widehat{\TTE} &= 
            \tfrac{q}{np} \sum\limits_{t=0}^{\beta} h_{t,q} \sum\limits_{i=1}^{n} \sum_{\cS \in \cS_i^{\beta}} c_{i,\cS} \prod_{j \in \cS} z_j^t \\
            &= \sum\limits_{t=0}^{\beta} h_{t,q} \cdot \bigg( \tfrac{1}{n} \sum\limits_{i=1}^{n} \sum_{\cS \in \cS_i^{\beta}} \tfrac{q}{p} \cdot c_{i,\cS} \cdot \Ind(\cS \subseteq \cU) \prod_{j \in \cS} \tilde{z}_j^t \bigg) \\
            &= \sum\limits_{t=0}^{\beta} h_{t,q} \cdot \bigg( \tfrac{1}{n} \sum\limits_{i=1}^{n} \sum_{\cS \in \cS_i^{\beta}} \tilde{c}_{i,\cS} \prod_{j \in \cS} \tilde{z}_j^t \bigg) \\
            &= \sum\limits_{t=0}^{\beta} h_{t,q} \cdot  \Big( \tfrac{1}{n} \sum\limits_{i=1}^{n} \tilde{Y}_i(\tilde{\mathbf{z}}^t) \Big),
    \end{align*}
    where
    \[
        \tilde{c}_{i,\cS} = \tfrac{q}{p} c_{i,\cS} \cdot \Ind(\cS \subseteq \cU),
        \hspace{40pt}
        \tilde{Y}_i(\tilde{\mathbf{z}}) = \sum_{\cS \in \cS_i^{\beta}} \tilde{c}_{i,\cS} \prod_{j \in \cS} \tilde{z}_j^t.
    \]
    Writing it in this way, we can see that the distribution of $\widehat{\TTE}$ conditioned on $\cU$ is equivalent to the distribution of the polynomial interpolation estimator in \citet{cortez2022neurips} with $\tilde{\mathbf{z}} \sim \text{CRD}(qn,n)$ for a modified potential outcomes model given by the coefficients $\tilde{c}_{i,\cS}$.
    
    Under the assumption that $c_{i,\cS} \geq 0$, then $\tilde{Y}_i(\bz) \leq \frac{q}{p} Y_{i}(\bz)$.
    
    As a result, the variance of $\widehat{\TTE}$ conditioned on $\cU$ can be upper-bounded from \eqref{eq:refined_rollour_var}. As this expression does not depend on $\cU$,
    \[
        \E_{\cU}\bigg[\Var_{\bz|\cU}\Big(\widehat{\TTE}\Big)\bigg] \leq \tfrac{ q^3 \beta^2 Y_{\max}^2}{ p^2 n} \cdot \Big( \tfrac{\beta}{q} \Big)^{2\beta} \cdot \Big( d^2 + 4\beta^3 \Big).
    \]

    \cycomment{Also our upper bounds were really under an assumption that the budget was less than 1/2, so we may need to double-check the proof for cases when $q > 1/2$.} \mecomment{I can't find anywhere where we needed this assumption... Do you have something specific in mind? I am worried though that we need to make explicit this new bounding assumption, which can be summarized by $p \geq \frac{2\beta^3+\beta}{n}$.}

    \underline{Second Term:}

    \meedit{First, let us note that when $q=p$, every individual is deterministically included in $\cU$ during Stage 1 of the experiment. In this case, the second term, which concerns a variance over $\cU$, is $0$. Thus, we may restrict our attention to the case when $q > p$ and multiply the resulting expression by the indicator $\Ind(q > p)$ in our final bound.}

    We first split $\E_{\bz} \big[ \hTTE \:\big|\: \cU \big] $ from \eqref{eq:e_tte_cond_u} into the terms associated to sets $\cS$ that are fully contained inside a cluster as opposed to sets $\cS$ that contain members of more than one cluster. 
    \begin{equation} \label{eq:var2terms}
    \E_{\bz} \big[ \hTTE \:\big|\: \cU \big] 
        = \tfrac{q}{np} \sum\limits_{i=1}^{n} \sum_{\cS \in \cS_i^\beta \setminus \varnothing} c_{i,\cS} \cdot \Ind \big( \cS \subseteq \cU, |\Pi(\cS)| = 1 \big) 
        + \tfrac{q}{np} \sum\limits_{i=1}^{n} \sum_{\cS \in \cS_i^\beta \setminus \varnothing} c_{i,\cS} \cdot \Ind \big( \cS \subseteq \cU, |\Pi(\cS)| \geq 2 \big).
    \end{equation}
    We may rewrite the first term of \eqref{eq:var2terms}:
    \[
    \tfrac{q}{np} \sum\limits_{i=1}^{n} \sum_{\cS \in \cS_i^\beta \setminus \varnothing} c_{i,\cS} \cdot \Ind \big( \cS \subseteq \cU, |\Pi(\cS)| = 1 \big)
        \quad=\quad \tfrac{q}{np} \sum_{i=1}^{n} \sum_{\pi \in \Pi} x_\pi \sum_{\substack{\cS \in \cS_i^\beta \\ \cS \ne \varnothing}} c_{i,\cS} \cdot \Ind \big(S \subseteq \pi \big) \\
        \quad=\quad \tfrac{q}{p n_c} \sum_{\pi \in \Pi} x_\pi \bar{L}_\pi,
    \]
    where $x_\pi = \Ind\big( \pi \subseteq \cU \big)$ and $\bar{L}_\pi$ is defined as in the main text, with 
    \[
        \bar{L}_\pi = \tfrac{n_c}{n} \sum_{i=1}^{n} \sum_{\cS \subseteq [n]} c_{i,\cS} \cdot \Ind \big(\cS \subseteq \pi \big),
    \]
    which represents the effects associated with sets fully contained inside cluster $\pi$. In Stage 1, we select clusters according to a CRD design. In particular, $\bx \sim \textrm{CRD}(p n_c/q, n_c)$. Applying Lemma \ref{lemma:CRD_PNAS}, we find that the variance of the first term of \eqref{eq:var2terms} is equal to
    \[
        \tfrac{q-p}{p(n_c -1)} \cdot \widehat{\Var}\big(\bar{L}_\pi \big).
    \]

    To upper bound the terms of the variance associated to the second term of $\E_{\bz} \big[ \hTTE \:\big|\: \cU \big]$ associated to all the sets $\cS$ for which $|\Pi(\cS)| \geq 2$, we use the bound that for any $\cS$ such that $|\Pi(\cS)| \geq 2$,
    \[
        \Cov\Big(\Ind \big( \cS \subseteq \cU), \Ind \big( \cS' \subseteq \cU)\Big) \leq \Pr(\cS \subseteq \cU) \cdot \Ind\Big(\Pi(\cS) \cap \Pi(\cS') \neq \emptyset \Big).
    \]
    In addition, we'll make use of our assumption that each  $c_{i\cS} \geq 0$. Plugging in these bounds, it follows that
    \begin{align*}
    &\Cov\Big(\tfrac{q}{np} \sum\limits_{i=1}^{n} \sum_{\cS \in \cS_i^\beta \setminus \varnothing} c_{i,\cS} \cdot \Ind \big( \cS \subseteq \cU, |\Pi(\cS)| \geq 2 \big), \tfrac{q}{p n_c} \sum_{\pi \in \Pi} x_\pi \bar{L}_\pi \Big) \\
    &= \tfrac{q^2}{n n_c p^2} \sum\limits_{i=1}^{n} \sum_{\cS \in \cS_i^\beta \setminus \varnothing} c_{i,\cS} \cdot \Ind \big(|\Pi(\cS)| \geq 2 \big) \sum_{\pi \in \Pi} \bar{L}_\pi \Cov (\Ind \big( \cS \subseteq \cU), x_\pi) \\
    &\leq \tfrac{q^2}{n n_c p^2} \sum\limits_{i=1}^{n} \sum_{\cS \in \cS_i^\beta \setminus \varnothing} c_{i,\cS} \cdot \Ind \big(|\Pi(\cS)| \geq 2 \big) \sum_{\pi \in \Pi} \bar{L}_\pi \cdot \Pr(\cS \subseteq \cU) \cdot \Ind(c \in \Pi(\cS)) \\
    &= \tfrac{q^2}{p^2 n n_c} \sum\limits_{i=1}^{n} \sum_{\cS \in \cS_i^\beta \setminus \varnothing} c_{i,\cS} \cdot \Pr(\cS \subseteq \cU) \cdot \Ind \big(|\Pi(\cS)| \geq 2 \big) \sum_{\pi \in \Pi(\cS)} \bar{L}_\pi \\
    &= \tfrac{q^2}{p^2 n^2} \sum\limits_{i=1}^{n} \sum_{\cS \in \cS_i^\beta \setminus \varnothing} c_{i,\cS} \cdot \Pr(\cS \subseteq \cU) \cdot \Ind \big(|\Pi(\cS)| \geq 2 \big) \sum_{\pi \in \Pi(\cS)} \sum_{i' \in [n]} \sum_{\cS' \in \cS_i'^\beta} c_{i',\cS'} \cdot \Ind \big(\cS' \subseteq \pi \big) \\
    &\leq \tfrac{q^2}{p^2 n^2} \sum\limits_{i=1}^{n} \sum_{\cS \in \cS_i^\beta \setminus \varnothing} c_{i,\cS} \cdot \Pr(\cS \subseteq \cU) \cdot \Ind \big(|\Pi(\cS)| \geq 2 \big) \sum_{\pi \in \Pi(\cS)} \sum_{i' \in [n]} \Ind \big(\pi \in \Pi(\cN_{i'}) \big) \sum_{\cS' \in \cS_i'^\beta} c_{i',\cS'} \\
    &\leq \tfrac{q^2 Y_{\max}}{p^2 n^2} \sum\limits_{i=1}^{n} \sum_{\cS \in \cS_i^\beta \setminus \varnothing} c_{i,\cS} \cdot \Pr(\cS \subseteq \cU) \cdot \Ind \big(|\Pi(\cS)| \geq 2 \big) \sum_{\pi \in \Pi(\cS)} \sum_{i' \in [n]} \Ind \big(\pi \in \Pi(\cN_{i'}) \big) \\
    &\leq \tfrac{q^2 Y_{\max}}{p^2 n^2} \sum\limits_{i=1}^{n} \sum_{\cS \in \cS_i^\beta \setminus \varnothing} c_{i,\cS} \cdot \Pr(\cS \subseteq \cU) \cdot \Ind \big(|\Pi(\cS)| \geq 2 \big) \sum_{\pi \in \Pi(\cS)} \tfrac{nd}{n_c} \\
    &= \tfrac{q^2 d \beta Y_{\max}}{p^2 n_c} \bigg( \tfrac{1}{n} \sum\limits_{i=1}^{n} \sum_{\cS \in \cS_i^\beta \setminus \varnothing} c_{i,\cS} \cdot \Pr(\cS \subseteq \cU) \cdot \Ind \big(|\Pi(\cS)| \geq 2 \big) \bigg).
    \end{align*}
    In addition, 
    \begin{align*}
        &\Var\Big(\tfrac{q}{np} \sum\limits_{i=1}^{n} \sum_{\cS \in \cS_i^\beta \setminus \varnothing} c_{i,\cS} \cdot \Ind \big( \cS \subseteq \cU, |\Pi(\cS)| \geq 2 \big) \Big) \\
        &= \tfrac{q^2}{n^2p^2} \sum\limits_{i=1}^{n} \sum_{\cS \in \cS_i^\beta \setminus \varnothing} c_{i,\cS} \cdot \Ind \big(|\Pi(\cS)| \geq 2 \big) \sum\limits_{i'=1}^{n} \sum_{\cS' \in \cS_{i'}^\beta \setminus \varnothing} c_{i',\cS'} \cdot \Ind \big(|\Pi(\cS')| \geq 2 \big) \cdot \Cov\Big(\Ind \big( \cS \subseteq \cU) , \Ind \big( \cS' \subseteq \cU)\Big) \\
        &\leq \tfrac{q^2}{n^2p^2} \sum\limits_{i=1}^{n} \sum_{\cS \in \cS_i^\beta \setminus \varnothing} c_{i,\cS} \cdot \Ind \big(|\Pi(\cS)| \geq 2 \big) \sum\limits_{i'=1}^{n} \sum_{\cS' \in \cS_{i'}^\beta \setminus \varnothing} c_{i',\cS'} \cdot  \Pr(\cS \subseteq \cU) \cdot \Ind\Big(\Pi(\cS) \cap \Pi(\cS') \neq \emptyset\Big) \\
        &\leq \tfrac{q^2}{p^2 n^2} \sum\limits_{i=1}^{n} \sum_{\cS \in \cS_i^\beta \setminus \varnothing} c_{i,\cS} \cdot \Pr(\cS \subseteq \cU) \cdot \Ind \big(|\Pi(\cS)| \geq 2 \big) \sum\limits_{i'=1}^{n} \Ind\Big(\Pi(\cN_i) \cap \Pi(\cN_i') \neq \emptyset\Big) \sum_{\cS' \in \cS_{i'}^\beta \setminus \varnothing} c_{i',\cS'}  \\ 
        &\leq \tfrac{q^2 Y_{\max}}{p^2 n} \bigg( \tfrac{1}{n} \sum\limits_{i=1}^{n} \sum_{\cS \in \cS_i^\beta \setminus \varnothing} c_{i,\cS} \cdot \Pr(\cS \subseteq \cU) \cdot \Ind \big(|\Pi(\cS)| \geq 2 \big) \sum\limits_{i'=1}^{n} \Ind\Big(\Pi(\cN_i) \cap \Pi(\cN_i') \neq \emptyset\Big) \bigg) \\ 
        &\leq \tfrac{q^2 d^2 Y_{\max}}{p^2 \mereplace{n}{n_c}} \left(\frac{1}{n} \sum\limits_{i=1}^{n} \sum_{\cS \in \cS_i^\beta \setminus \varnothing} c_{i,\cS} \Pr(\cS \subseteq \cU) \Ind \big(|\Pi(\cS)| \geq 2 \big) \right).
    \end{align*}

    Putting it all together, we get that 
    \mcreplace{
    \mereplace{
    \begin{align*}
        \Var_{\cU}\bigg[\E_{\bz}\Big(\hTTE \ | \ \cU\Big)\bigg] \leq \tfrac{q-p}{p(n_c -1)} \cdot \widehat{\Var}\big(\bar{L}_\pi \big) 
            + \left(\tfrac{Y_{\max} \beta (d+1)}{n_c} + \tfrac{Y_{\max} d^2}{n}\right) \left(\frac{q^2}{p^2 n} \sum\limits_{i=1}^{n} \sum_{\cS \in \cS_i^\beta \setminus \varnothing} c_{i,\cS} \Pr(\cS \subseteq \cU) \Ind \big(|\Pi(\cS)| \geq 2 \big) \right).
    \end{align*}
    }
    {
    \begin{align*}
        \Var_{\cU}\bigg[\E_{\bz}\Big(\hTTE \ | \ \cU\Big)\bigg] 
        &\leq \tfrac{q-p}{p(n_c -1)} \cdot \widehat{\Var}\big(\bar{L}_\pi \big) 
            + \tfrac{2 d^2 Y_{\max}}{n_c} \bigg(\tfrac{q^2}{p^2 n} \sum\limits_{i=1}^{n} \sum_{\cS \in \cS_i^\beta \setminus \varnothing} c_{i,\cS} \cdot \Pr(\cS \subseteq \cU) \cdot \Ind \big(|\Pi(\cS)| \geq 2 \big) \bigg) \\
        &\leq \tfrac{q-p}{p(n_c -1)} \cdot \widehat{\Var}\big(\bar{L}_\pi \big) 
            + \tfrac{2 d^2 Y_{\max}}{n_c} \bigg(\tfrac{1}{n} \sum\limits_{i=1}^{n} \sum_{\cS \in \cS_i^\beta \setminus \varnothing} c_{i,\cS} \cdot \Ind \big(|\Pi(\cS)| \geq 2 \big) \bigg) \\
        &\leq \tfrac{q-p}{p(n_c -1)} \cdot \widehat{\Var}\big(\bar{L}_\pi \big) 
            + \tfrac{2 d^2 Y_{\max}}{n_c} \cdot C(\delta(\Pi)).
    \end{align*}
    Here, the second inequality uses the fact that 
    \[
        \Pr(\cS \subseteq \cU) \cdot \Ind \big(|\Pi(\cS)| \geq 2 \big) \leq (p/q)^2 \cdot \Ind \big(|\Pi(\cS)| \geq 2 \big).
    \]
    }}
    {\begin{align*}
        \Var_{\cU}\bigg[\E_{\bz}\Big(\hTTE \ | \ \cU\Big)\bigg] 
        &\leq \tfrac{q-p}{p(n_c -1)} \cdot \widehat{\Var}\big(\bar{L}_\pi \big) 
            + \Big(\tfrac{d \beta Y_{\max}}{n_c} + \tfrac{d^2 Y_{\max}}{ n_c}\Big) \bigg(\tfrac{q^2}{p^2 n} \sum\limits_{i=1}^{n} \sum_{\cS \in \cS_i^\beta \setminus \varnothing} c_{i,\cS} \cdot \Pr(\cS \subseteq \cU) \cdot \Ind \big(|\Pi(\cS)| \geq 2 \big) \bigg)\\
        &\leq \tfrac{q-p}{p(n_c -1)} \cdot \widehat{\Var}\big(\bar{L}_\pi \big) 
            + \Big(\tfrac{d \beta}{n_c} + \tfrac{d^2}{n_c}\Big) Y_{\max} \bigg(\tfrac{1}{n} \sum\limits_{i=1}^{n} \sum_{\cS \in \cS_i^\beta \setminus \varnothing} c_{i,\cS} \cdot \Ind \big(|\Pi(\cS)| \geq 2 \big) \bigg) \\
        &\leq \tfrac{q-p}{p(n_c -1)} \cdot \widehat{\Var}\big(\bar{L}_\pi \big) 
            + \Big(\tfrac{d \beta}{n_c} + \tfrac{d^2}{n_c}\Big) \cdot Y_{\max} C(\delta(\Pi))\\
        &\leq \tfrac{q-p}{p(n_c -1)} \cdot \widehat{\Var}\big(\bar{L}_\pi \big) 
            + \tfrac{2 d^2}{n_c} \cdot Y_{\max} C(\delta(\Pi)).
    \end{align*}}
    Here, the second inequality uses the fact that 
    \[
        \Pr(\cS \subseteq \cU) \cdot \Ind \big(|\Pi(\cS)| \geq 2 \big) \leq (p/q)^2 \cdot \Ind \big(|\Pi(\cS)| \geq 2 \big).
    \]
    \end{proof}

\subsection*{Theorem \ref{thm:variance-deg1-crd}}

\begin{proof}
When $\beta = 1$, the estimator simplifies to
\[
    \widehat{\TTE}
    = \tfrac{1}{np} \sum_{i \in [n]} \Big(Y_i(\bz^1) - Y_i(\mathbf{0})\Big)
    = \tfrac{1}{np} \sum_{j \in [n]} L_j z_j^1
    = \tfrac{1}{np} \sum_{\pi} \sum_{j \in \pi} L_j z_j^1.
\]
Conditioning on $\cU$, the estimator becomes
\[
   \tfrac{1}{np} \sum_{j \in \cU} L_j z_j^1 = \tfrac{1}{q |\cU|} \sum_{j \in \cU} L_j z_j^1,
\]
where here we use the fact that $|\cU| = \tfrac{np}{q}$. Since $\bz_\cU \sim \textrm{CRD}(q|\cU|,q)$, we may use Lemma~\ref{lemma:CRD_PNAS} to obtain an expression for the conditional variance:
\[
 \Var_{\bz\mid\cU} \big( \hTTE \big) = \tfrac{1-q}{q(|\cU|-1)} \cdot \bigg[\tfrac{1}{|\cU|} \sum_{j\in\cU} L_j^2 \ - \ \Big( \tfrac{1}{|\cU|} \sum_{j\in\cU} L_j \Big)^2 \bigg].
\]
Taking the expectation of this conditional variance with respect to $\cU$, we have
\begin{align*}
    \E_\cU\Big[\Var_{\bz\mid\cU} \big( \hTTE \big)\Big] 
    &= \tfrac{1-q}{q(|\cU|-1)} \bigg[ \tfrac{1}{|\cU|} \sum_{j \in [n]}L_j^2 \cdot \Pr(j\in\cU) \ - \ \tfrac{1}{|\cU|^2} \sum_{j \in [n]} \sum_{j' \in [n]} L_jL_{j'} \cdot \Pr(j,j'\in\cU) \bigg] \\
    &= \tfrac{1-q}{np-q} \bigg[ \tfrac{q}{np} \sum_{j \in [n]} L_j^2 \cdot \Pr(j\in\cU) \ - \ \tfrac{q^2}{n^2p^2} \sum_{j \in [n]} \sum_{j' \in [n]} L_jL_{j'} \cdot \Pr(j,j'\in\cU) \bigg] \\
    &= \tfrac{1-q}{np-q} \bigg[ \tfrac{1}{n} \sum_{j \in [n]} L_j^2 
        - \tfrac{q}{n^2p} \sum_{j \in [n]} \sum_{j' \in \pi(j)} L_jL_{j'}
        - \tfrac{p n_c - q}{n^2p(n_c-1)} \sum_{j \in [n]} \sum_{j' \not\in \pi(j)} L_jL_{j'} \bigg] \\
    &= \tfrac{1-q}{np-q} \bigg[ \tfrac{1}{n} \sum_{j \in [n]}L_j^2
        - \tfrac{q}{n^2p} \sum_{\pi \in \Pi} \Big( \sum_{j \in \pi} L_j \Big)^2 
        - \tfrac{p n_c - q}{n^2p(n_c-1)} \Big[ \Big( \sum_{\pi \in \Pi} \sum_{j \in \pi}  L_j \Big)^2 - \sum_{\pi \in \Pi} \Big( \sum_{j \in \pi} L_j \Big)^2 \Big] \bigg] \\
    &= \tfrac{1-q}{np-q} \bigg[ \tfrac{1}{n} \sum_{j \in [n]}L_j^2
        - \tfrac{q}{n^2p} \sum_{\pi \in \Pi} \Big( \sum_{j \in \pi} L_j \Big)^2 
        - \tfrac{p n_c - q}{n^2p(n_c-1)} \Big( \sum_{j \in [n]}  L_j \Big)^2 
        + \tfrac{p n_c - q}{n^2p(n_c-1)} \sum_{\pi \in \Pi} \Big( \sum_{j \in \pi} L_j \Big)^2 \bigg] \\
    &= \tfrac{1-q}{np-q} \bigg[ \tfrac{1}{n} \sum_{j \in [n]}L_j^2
        + \tfrac{(p-q)n_c}{n^2p(n_c-1)} \sum_{\pi \in \Pi} \Big( \sum_{j \in \pi} L_j \Big)^2 
        - \tfrac{p n_c - q}{p(n_c-1)} \Big( \tfrac{1}{n} \sum_{j \in [n]}  L_j \Big)^2 \bigg] \\
    &= \tfrac{1-q}{np-q} \bigg[ \Big[ \tfrac{1}{n} \sum_{j \in [n]} L_j^2 - \Big( \tfrac{1}{n} \sum_{j \in [n]}  L_j \Big)^2 \Big]
        + \tfrac{p-q}{p(n_c-1)} \cdot \tfrac{1}{n_c} \sum_{\pi \in \Pi} \Big( \tfrac{n_c}{n} \sum_{j \in \pi} L_j \Big)^2 
        - \tfrac{p-q}{p(n_c-1)} \Big( \tfrac{1}{n_c} \sum_{\pi \in \Pi} \tfrac{n_c}{n} \sum_{j \in \pi}  L_j \Big)^2 \bigg] \\
    &= \tfrac{1-q}{np-q} \bigg[ \; \underset{j \in [n]}{\widehat{\Var}} \big( L_j \big)
        + \tfrac{p-q}{p(n_c-1)} \Big[ \tfrac{1}{n_c} \sum_{\pi \in \Pi} \big( \bar{L}_\pi \big)^2 - \Big( \tfrac{1}{n_c} \sum_{\pi \in \Pi} \bar{L}_\pi \Big)^2 \Big] \bigg] \\
    &= \tfrac{1-q}{np-q} \bigg[ \; \underset{j \in [n]}{\widehat{\Var}} \big( L_j \big)
        + \tfrac{p-q}{p(n_c-1)} \cdot \underset{\pi \in \Pi}{\widehat{\Var}} \big( \bar{L}_\pi \big) \bigg].
\end{align*}

The conditional expectation is given by 
\[
    \E_{\bz} \big[ \widehat{\TTE} \;\big|\; \cU \big] 
    = \tfrac{q}{pn_c} \sum_{\pi\in\Pi} \Big(\tfrac{n_c}{n}\sum_{j\in\pi} L_j\Big) \cdot \Ind(\pi \subseteq \cU)
    = \tfrac{q}{pn_c} \sum_{\pi\in\Pi} \bar{L}_j \cdot \Ind(\pi \subseteq \cU).
\]
Since these indicator random variables are sampled in Stage 1 according to a $\textrm{CRD}(pn_c/q, n_c)$ distribution, we may apply Lemma~\ref{lemma:CRD_PNAS} to conclude that
\[
    \Var\Big(\E_{\bz\mid\cU}\big[\hTTE\big]\Big) 
    = \tfrac{1 - (p/q)}{(p/q)(n_c-1)} \cdot \underset{\pi \in \Pi}{\widehat{\Var}} \big(\bar{L}_\pi\big)
    = \tfrac{q-p}{p(n_c-1)} \cdot \underset{\pi \in \Pi}{\widehat{\Var}} \big(\bar{L}_\pi\big).
\]

Putting this together, we find that
\[
    \Var\Big( \widehat{\TTE} \Big)
    = \tfrac{1-q}{np-q} \cdot \underset{j \in [n]}{\widehat{\Var}} \big( L_j \big)
        + \tfrac{(p-q)(1-np)}{p(np-q)(n_c-1)} \cdot \underset{\pi \in \Pi}{\widehat{\Var}} \big( \bar{L}_\pi \big).
\]
\end{proof}

\section{Experiment Details} 
\label{app:experiments}

\subsection{Potential Outcomes Model}
We generate synthetic potential outcomes based on a generalization of the response model from \citet{ugander2023randomized} to incorporate $\beta$-order interactions:
\[
    Y_i(\bz) = Y_i(\mathbf{0}) \cdot \bigg( 1 + \delta z_i + \sum_{k=1}^{\beta} \gamma_k \cdot \Big( \begin{matrix} d_i \\ k \end{matrix} \Big)^{-1} \sum_{\substack{\cS \in \cS_i^\beta \\ |\cS| = k}} \prod_{j \in \cS} z_j \bigg), \quad \quad
    Y_i(\mathbf{0}) = \Big( a + b \cdot h_i + \varepsilon_i \Big) \cdot \tfrac{d_i}{\overline{d}}.
\]

In this model: 
\begin{itemize}
    \item $a$ is a baseline effect. We select $a=1$.
    \item $\textbf{h} \in \mathbb{R}^n$ is a Fiedler vector of the graph Laplacian of the network which has undergone an affine transformation so that $\min(\textbf{h})=-1$ and $\max(\textbf{h})=1$. This models network homophily effects. 
    \item $b$ controls the magnitude of the homophily effect. We select \mcreplace{$b=0.5$}{$b=0$}. \mcedit{We also ran the experiments with $b=0.5$, to compare no homophily with some homophily, but the analysis and conclusions do not change. These are included later in the appendix. }
    \item $\varepsilon_i \underset{\textrm{iid}}{\sim} N(0,\sigma)$ is a random perturbation of the baseline effect. We select $\sigma = 0.1$.
    \item $d_i$ is the in-degree of vertex $i$. $\overline{d}$ is the average in-degree.
    \item $\delta$ is uniform direct effect on treated individuals. We select $\delta = 0.5$.
    \item $\gamma_k$ is the effect of treated subsets of size $k$. We select $\gamma_k = 0.5^{k-1}$, which models marginal effects that decay with the size of the treated set. 
\end{itemize}

\subsection{Details of Real-World Networks} 
\label{app:networks}

Here, we provide more details of the three real-world data sets we use in our analysis. We include all the raw data files, cleaned data, and processing scripts in our provided source code. A summary of the datasets is given in the following table.

\begin{center}
    \begin{tabular}{|c|c|c|c|c|} \hline
        Dataset & Vertices & Edges & Degree & Features \\ \hline \hline
        \makecell{\textsc{Email} \\ \cite{leskovec2007graph,yin2017local};\\ \cite{snapnets}} & \makecell{employees \\ $n=1,005$} & \makecell{correspondence \\ directed \\ $|E|=25,571$} & \makecell{min: $1$ \\ max: $334$ \\ average: $25$} & \makecell{department \\ $|F| = 42$} \\ \hline
        \makecell{\textsc{BlogCatalog} \\ \cite{network_repo};\\ \cite{tang2009scalable,tang2009relational}} & \makecell{bloggers \\ $n=10,312$} & \makecell{connections \\ undirected \\ $|E|=333,983$} & \makecell{min: $1$ \\ max: $3,992$ \\ average: $65$} & \makecell{interests \\ $|F|=39$} \\ \hline
        \makecell{\textsc{Amazon} \\ \cite{leskovec2007dynamics};\\ \cite{snapnets}} &  \makecell{products \\ $n=14,436$} & \makecell{co-purchases \\ directed \\ $|E| = 70,832$} & \makecell{min: $1$ \\ max: $247$ \\ average: $5$} & \makecell{category \\ $|F|=13,591$} \\ \hline
    \end{tabular}
\end{center}

\subsection*{\textsc{Email}}

The \textsc{Email} dataset is publicly available at \url{https://snap.stanford.edu/data/email-Eu-core.html} and is licensed under the BSD license\footnote{For more information, see \url{https://snap.stanford.edu/snap/license.html} and \url{https://groups.google.com/g/snap-datasets/c/52MRzGbMkFg/m/FIFy_6qOCAAJ}}. This dataset models the email communications between members of a European research institution. The $n=1,005$ vertices of the network are (anonymized) institution members, and there is a directed edge from individual $i$ to individual $j$ if $i$ has sent at least one email to individual $j$. 

It has a minimum degree of 1, a maximum degree of 212, and an average degree of 25.8, and its degree distribution is visualized in Figure~\ref{fig:email_degrees}; the support of the histogram has been cropped to remove some large outliers. The largest weakly connected component in the network contains 986 vertices, and the largest strongly connected component contains 803 vertices. 

\begin{figure}[h]
    \centering
    \includegraphics[width=0.8\textwidth]{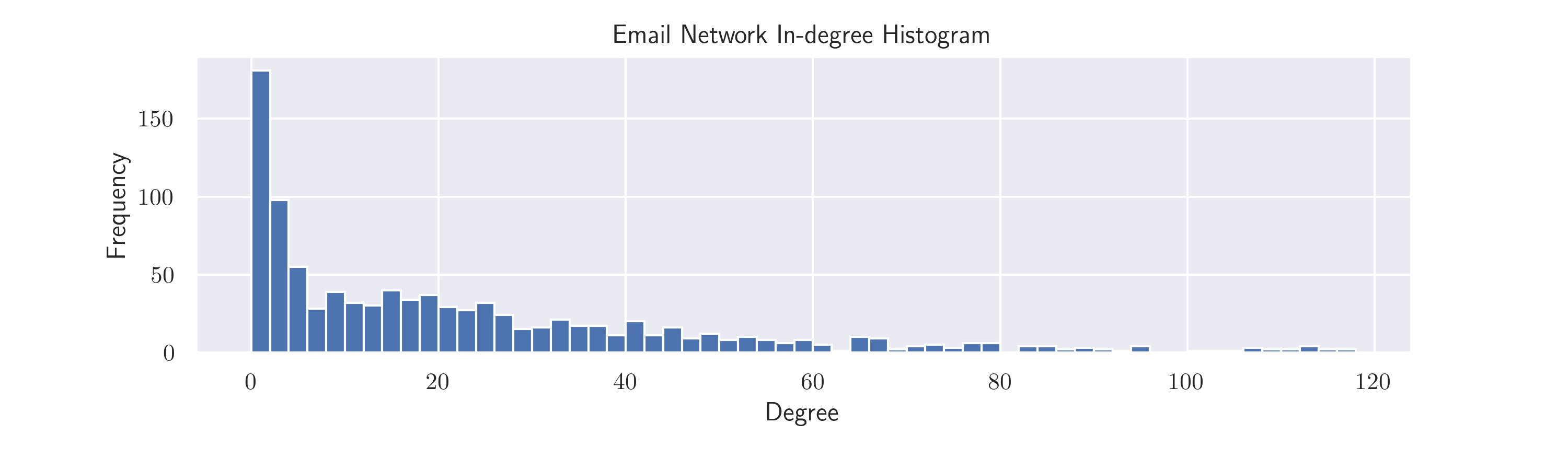}
    \caption{The degree distribution of the \textsc{Email} graph}
    \label{fig:email_degrees}
\end{figure}

Each individual in the network has been assigned one of 42 department labels. The sizes of these departments vary greatly, with the smallest department including a single individual and the largest department including 109 individuals. The average department size is 23.9. In the \textsc{Email} network, each vertex is assigned to exactly one department, and we use these assignments as our clustering. To pre-process this data for use in our experiments, we added self-loops to each node in the original dataset to represent the direct effect of the node's treatment on their outcome (See Section \ref{sec:preliminaries}).

\subsection*{\textsc{BlogCatalog}}
The \textsc{BlogCatalog} dataset is publicly available at \url{https://networkrepository.com/soc-BlogCatalog-ASU.php} and is licensed under a Creative Commons Attribution-ShareAlike License\footnote{For more information, see \url{https://networkrepository.com/policy.php}}. This dataset models the relationships between bloggers on the (now defunct) blogging website \textsf{http://www.blogcatalog.com}. The $n=10,312$ nodes represent bloggers and the (undirected) edges represent the social network of the bloggers. 

The network has a minimum degree of 1, a maximum degree of 3,992, and an average degree of 65, and its degree distribution is visualized in Figure \ref{fig:blog_degrees}; the support of the histogram has been cropped to remove some large outliers. The average clustering coefficient is approximately 0.46. 

\begin{figure}[h]
    \centering
    \includegraphics[width=0.8\textwidth]{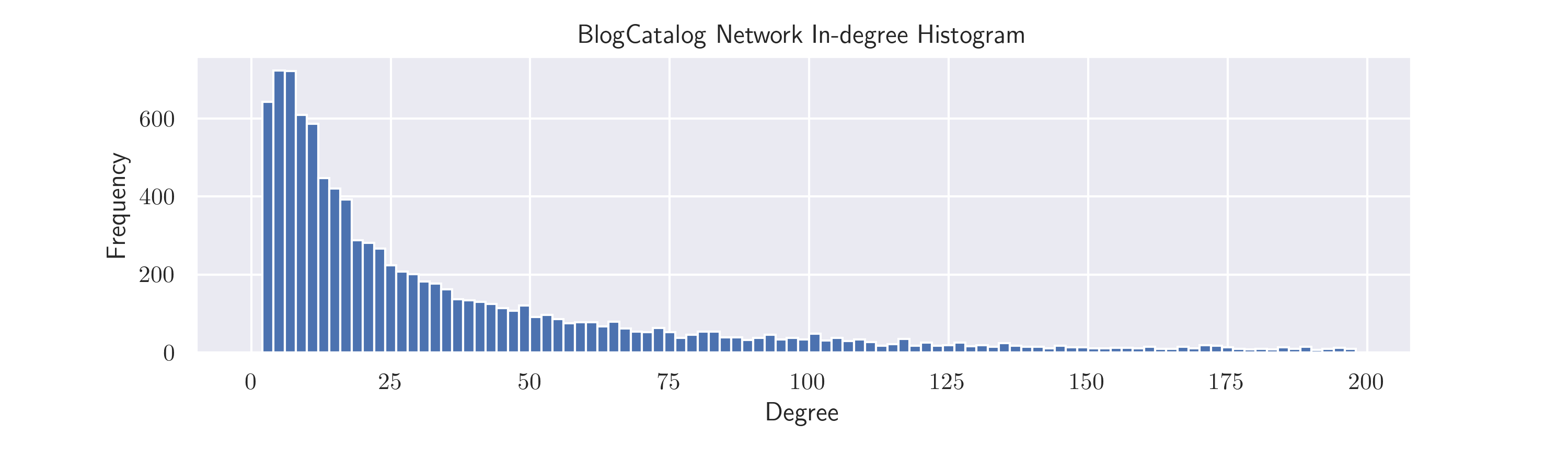}
    \caption{The degree distribution of the \textsc{BlogCatalog} graph}
    \label{fig:blog_degrees}
\end{figure}

Each blogger in the network has an associated blog. Blogs (and thus, bloggers) are organized under interest categories specified by the website and can be listed under multiple categories. There are 39 such categories in this dataset and on average; on average, each blogger is listed under 1.6 categories. As part of the data pre-processing for our experiments, we added self-loops to each node in the original dataset, as we did with the \textsc{Email} dataset.

\subsection*{\textsc{Amazon}}
The \textsc{Amazon} dataset is publicly available at \url{https://snap.stanford.edu/data/amazon-meta.html} and is licensed under the BSD license\footnote{For more information, see \url{https://snap.stanford.edu/snap/license.html} and \url{https://groups.google.com/g/snap-datasets/c/52MRzGbMkFg/m/FIFy_6qOCAAJ}}. This dataset models an Amazon product co-purchasing network. The $n=14,436$ nodes represent products and each node has outgoing edges to the top 5 products with which it is a frequent co-purchase. Thus, in addition to the self-loop at each node, each node has exactly 5 outgoing edges.

The network has a minimum in-degree of 1, a maximum in-degree of 247, and an average in-degree of 5; its in-degree distribution is visualized in Figure \ref{fig:amazon_degrees}; the support of the histogram has been cropped to remove some large outliers. 
\begin{figure}[h]
    \centering
    \includegraphics[width=0.8\textwidth]{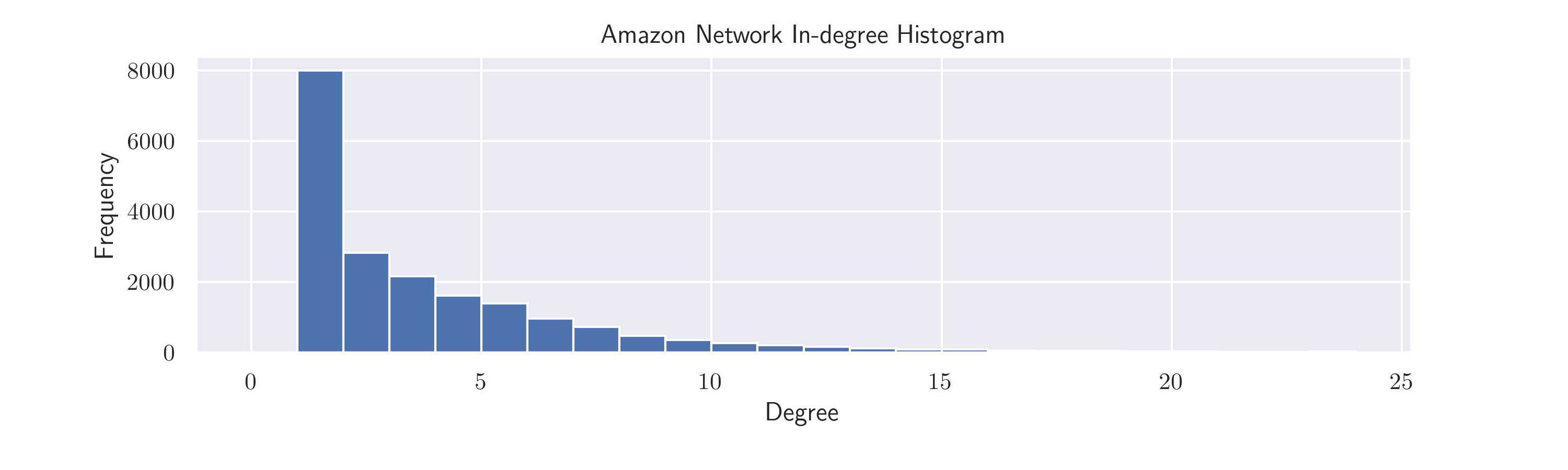}
    \caption{The degree distribution of the \textsc{Amazon} graph}
    \label{fig:amazon_degrees}
\end{figure}

Products are organized into categories (which correspond to attributes such as the genre, setting, and actors in the film, as well as marketplace data such as the inclusion of these titles in certain deals or promotions) but can belong to multiple categories. There are 13,591 possible product categories; on average each product belongs to 13.2 categories. As part of the data pre-processing for our experiments, we added all self-loops and restricted the original dataset to only the product nodes labeled as DVDs.

\subsection{Other Estimators} 
\label{app:other_ests}

Here, we provide additional details of the other estimators used in our experiments Section~\ref{sec:experiments}.

\subsection*{Difference-in-Means}

The difference-in-means (DM) approach estimates the total treatment effect as the difference in the average outcome of a treated individual and the average outcome of an untreated individual. 
\[
    \widehat{\TTE}_{\textrm{DM}} = \frac{\sum_{i} z_i \cdot Y_i(\bz)}{\sum_{i} z_i} - \frac{\sum_{i} (1-z_i) \cdot Y_i(\bz)}{\sum_{i} \;(1-z_i)}.
\]
This estimator does not utilize any knowledge of the underlying graph. It only requires knowledge of the treatment assignments and the observed outcomes, which are always known to the experimenter. Under SUTVA, this is an unbiased estimator (since an individual's outcome is a function only of their treatment). However, it is not unbiased under interference; untreated individuals may have treated neighbors which impacts their outcome, introducing bias into our signal of the baseline effects.\mcdelete{ In general, in settings where the treatment effects are positive. This phenomenon outweighs the bias from untreated neighbors of treated individuals, so we observe a negative bias.}

To counteract this bias, at the expense of requiring network knowledge, we can limit the set of individuals used in the estimator to those whose neighbors' treatment assignments largely align with theirs. We refer to this as a thresholded DM estimator.

\subsection*{Thresholded Difference-in-Means:} 

This family of estimators is parameterized by a value $\gamma \in [0,1]$, which can be viewed as a stringency requirement that we place on the treatment assignments within one's neighborhood. In particular, we only include an individual in the ``treated'' set in this DM estimator if they are treated and at least a $\gamma$ fraction of their neighbors are also treated.  Similarly, we will only include an individual in the ``untreated'' set in this DM estimator if they are untreated and at most a ($1-\gamma$) fraction of their neighbors are treated. 
\[
    \widehat{\TTE}_{\textrm{DM}(\gamma)} = \frac{\sum_{i} z_i \cdot Y_i(\bz) \cdot \Ind \big( \sum_{j \in \cN_i} z_i \geq \gamma d_i \big)}{\sum_{i} z_i \cdot \Ind \big( \sum_{j \in \cN_i} z_i \geq \gamma d_i \big)} 
    - \frac{\sum_{i} (1-z_i) \cdot Y_i(\bz) \cdot \Ind \big( \sum_{j \in \cN_i} z_i \leq (1-\gamma) d_i \big)}{\sum_{i} (1-z_i) \cdot \Ind \big( \sum_{j \in \cN_i} z_i \leq (1-\gamma) d_i \big)}
\]
Note that these estimators for $\gamma > 0$ require network knowledge to calculate the neighborhood treatment proportions, and they are biased under interference for the same reasoning as the standard DM estimator. Note that DM(0) (i.e., the thresholded DM estimator with parameter $\gamma = 0$) coincides with the ordinary DM estimator. The DM(1) estimator will only consider individuals with fully treated or untreated neighborhoods. As such, the DM(1) estimator will, under simpler randomization schemes like Bernoulli design, include very few individuals in its ``treated'' and ``untreated'' sets. The Horvitz-Thompson and H\'{a}jek estimators also exhibit this phenomenon.

\subsection*{Horvitz-Thompson:}

The Horvitz-Thompson (HT) estimator \citep{horvitz1952generalization} uses inverse probability weighting to construct an unbiased TTE estimator under \textit{arbitrary} potential outcomes models. To do this, it must only incorporate the outcomes from an individual's neighborhoods that are either fully treated or fully untreated, as these are the only outcomes that appear in the $\TTE$ estimand. The estimator has the form,
\begin{align*}
    \widehat{\TTE}_{\textrm{HT}} 
    &= \tfrac{1}{n} \sum_{i \in [n]} \frac{Y_i(\bz) \cdot \Ind(\cN_i \textrm{ fully treated})}{\Pr(\cN_i \textrm{ fully treated})}  - \tfrac{1}{n} \sum_{i \in [n]} \frac{Y_i(\bz) \cdot \Ind(\cN_i \textrm{ fully untreated})}{\Pr(\cN_i \textrm{ fully untreated})} \\
    &= \tfrac{1}{n} \sum_{i \in [n]} \bigg[ \tfrac{Y_i(\bz) \cdot \prod_{j \in \cN_i} z_j}{\Pr \big( \prod_{j \in \cN_i} z_j = 1 \big)} -\tfrac{Y_i(\bz) \cdot \prod_{j \in \cN_i} (1-z_j)}{\Pr \big( \prod_{j \in \cN_i} (1-z_j) = 1 \big)}
    \bigg] 
\end{align*}
This estimator is unbiased, but it relies on network knowledge to compute the exposure probabilities in the denominators of each fraction. A related inverse probability weighted estimator is the H\'{a}jek estimator.

\subsection*{H\'{a}jek}

Since the HT estimator only considers individuals with fully treated and fully untreated neighborhoods, most of the bracketed terms within its summation will be 0. To compensate for this, we can adjust the $1/n$ normalization on the summation to use the expected number of non-zero entries corresponding to both terms in the bracketed expression. This change gives the  H\'{a}jek estimator \citep{basu2011essay}.
\begin{align*}
    \widehat{\TTE}_{\textrm{H\'{a}jek}} = \frac{\sum\limits_{i \in [n]} \tfrac{ Y_i(\bz) \cdot \prod_{j \in \cN_i} z_j}{\Pr \big( \prod_{j \in \cN_i} z_j = 1 \big)}}{\sum\limits_{i \in [n]} \tfrac{\prod_{j \in \cN_i} z_j}{\Pr \big( \prod_{j \in \cN_i} z_j = 1 \big)}} 
    - \frac{\sum\limits_{i \in [n]} \tfrac{Y_i(\bz) \cdot \prod_{j \in \cN_i} (1-z_j)}{\Pr \big( \prod_{j \in \cN_i} (1-z_j) = 1 \big)}}{\sum\limits_{i \in [n]} \tfrac{\prod_{j \in \cN_i} (1-z_j)}{\Pr \big( \prod_{j \in \cN_i} (1-z_j) = 1 \big)}}
\end{align*}
The H\'{a}jek estimator trades off a reduction in the variance over the HT estimator for the introduction of some bias (a thorough discussion of this tradeoff is given by~\citet{khan2023adaptive}). As with the Horvitz-Thompson estimator, the calculation of exposure probabilities in this estimator requires knowledge of the interference network.

\subsection*{Two-Stage Estimator when $q=1$}

The one-stage estimator from \citet{cortez2022neurips} is 
\begin{equation} \label{eq:estimatorDefn-1Stage}
    \widehat{\TTE}_{\textrm{1-Stage}}^{\beta} := 
        \frac{1}{n} \sum\limits_{i=1}^{n}  \sum\limits_{t=0}^{\beta} \Big( \ell_{t,p}(1) - \ell_{t,p}(0) \Big) \cdot Y_i(\mathbf{z}^t),
        \quad
        \ell_{t,p}(x) = \prod_{\substack{s=0 \\ s\not=t}}^{\beta} \tfrac{\beta x-ps}{pt-ps}.
\end{equation}
When evaluating the estimator with $\beta=1$, the estimator is simply
\begin{equation} \label{eq:beta=1 estimator}
    \widehat{\TTE}_{\textrm{1-Stage}}^{\beta=1} = \frac{1}{np} \sum\limits_{i=1}^{n} \Big( Y_i(\bz^1) - Y_i(\bz^0) \Big)
\end{equation}
In what follows, we show that the two-stage rollout estimator with $q=1$ is equivalent to $\widehat{\TTE}_{\textrm{1-Stage}}^{\beta=1}$.

\begin{theorem} \label{thm:q=1,2-stage,equivalence}
    Under a Two-Stage Rollout Design with budget $p$ and effective treatment budget $q=1$, the two-stage estimator defined in equation \eqref{eq:estimatorDefn} under a model with degree $\beta$ is equivalent to the estimator defined in \eqref{eq:beta=1 estimator}.
\end{theorem}
\begin{proof}
    Under a Two-Stage Rollout Design with $q=1$, we have 
    \[
    h_{t,q} = h_{t,1} = \prod_{\substack{s=0 \\ s\not=t}}^{\beta} \tfrac{\beta-s}{t-s} \ - \ \prod_{\substack{s=0 \\ s\not=t}}^{\beta} \tfrac{-s}{t-s}.
    \]
Notice that when $t\in\{1,2,\ldots,\beta-1\}$, i.e. $t\ne0$ and $t\ne\beta$, at some point we have a term corresponding to $s=0$ and $s=\beta$ in the products above. Thus, both products are $0$.

When $t=0$, we have 
$
h_{0,1} \ = \ \prod_{\substack{s=1}}^{\beta} \tfrac{\beta-s}{0-s} \ - \ \prod_{\substack{s=1}}^{\beta} \tfrac{-s}{0-s} \ = \ \prod_{\substack{s=1}}^{\beta} \tfrac{\beta-s}{-s} \ - \ 1 \ = \ -1
$
because the product equals $0$ due to the $s=\beta$ term. Similarly, when $t=\beta$, we have 
$
h_{\beta,1} \ = \ \prod_{\substack{s=0}}^{\beta-1} \tfrac{\beta-s}{\beta-s} \ - \ \prod_{\substack{s=0}}^{\beta-1} \tfrac{-s}{\beta-s} \ = \ 1
$
since the second product will equal $0$ due to the $s=0$ term.
To summarize, we have
\[
h_{t,q} = 
\begin{cases}
    -1 & t = 0\\
    0 & 1 \leq t \leq \beta-1 \\
    1 & t = \beta
\end{cases}.
\]
\end{proof}

\subsection{Additional Experiments: Comparing Different Estimators}

In this section, we have figures showing the MSE of different estimators as we vary the treatment budget $p$ from $0.1$ to $0.5$ for different model degrees $\beta$ and different real-world networks. As a reminder, we compare the following estimators:
\begin{itemize}
    \item The two-stage polynomial interpolation estimator with $q=0.5$ and no clustering, \textsf{2-Stage}
    \item The two-stage polynomial interpolation estimator with $q=1$ and no clustering, \textsf{q=1}
    \item The one-stage polynomial interpolation estimator, \textsf{PI}, from \citet{cortez2022neurips}
    \item The simple difference-in-means estimator, \textsf{DM}
    \item The thresholded difference-in-means estimator with parameter $0.75$, \textsf{DM(0.75)}
    \item The H\'{a}jek estimator, \textsf{H\'{a}jek}
\end{itemize}
The first three estimators in this list are based on polynomial interpolation (PI), so we refer to them as the PI estimators. We refer to the others as the non-PI estimators. In all MSE plots, the lines indicate the empirical MSE over 1000 replications. In all bias and standard deviation plots, the bold line indicates the mean over 1000 replications, and the shading indicates the experimental standard deviation, calculated by taking the square root of the experimental variance over all replications. 

In Figure \ref{fig:amazon_est_mse}, we show the MSE corresponding to Figure \ref{fig:amazon} from Section \ref{sec:experiments}. The column faceting indicates model degree; note that the $y$-axis limits differ across these subplots. When $\beta=1$, \textsf{PI} and \textsf{q=1} are equivalent and have slightly lower MSE compared with \textsf{2-Stage}. However, the difference is hard to see without zooming in further since the lines almost overlap. Note that the difference-in-means estimators have MSE outside the bounds of the plots. When $\beta=2$, the results are similar but you can start to see the difference between the three PI estimators, which all have lower MSE when compared with the non-PI estimators. For smaller values of $p$, the estimator \textsf{PI} has slightly lower MSE, followed by \textsf{2-Stage}, followed by \textsf{q=1}. Again, the difference is quite small. In this case, as noted in the main body of the paper, we are in a setting where using the one-stage rollout and estimator is preferable.
When $\beta=3$, the difference between the three PI estimators is more pronounced. The \textsf{2-Stage} has the lowest MSE, especially for lower values of $p$. The \textsf{q=1} estimator has MSE relatively close to it for all $p$-values, but does slightly worse, although better than \textsf{PI} for small $p$ values. In this case, we have a setting where the two-stage approach is valuable as it outperforms the other methods.

\begin{figure}[h]
    \centering
    \includegraphics[width=0.8\linewidth]{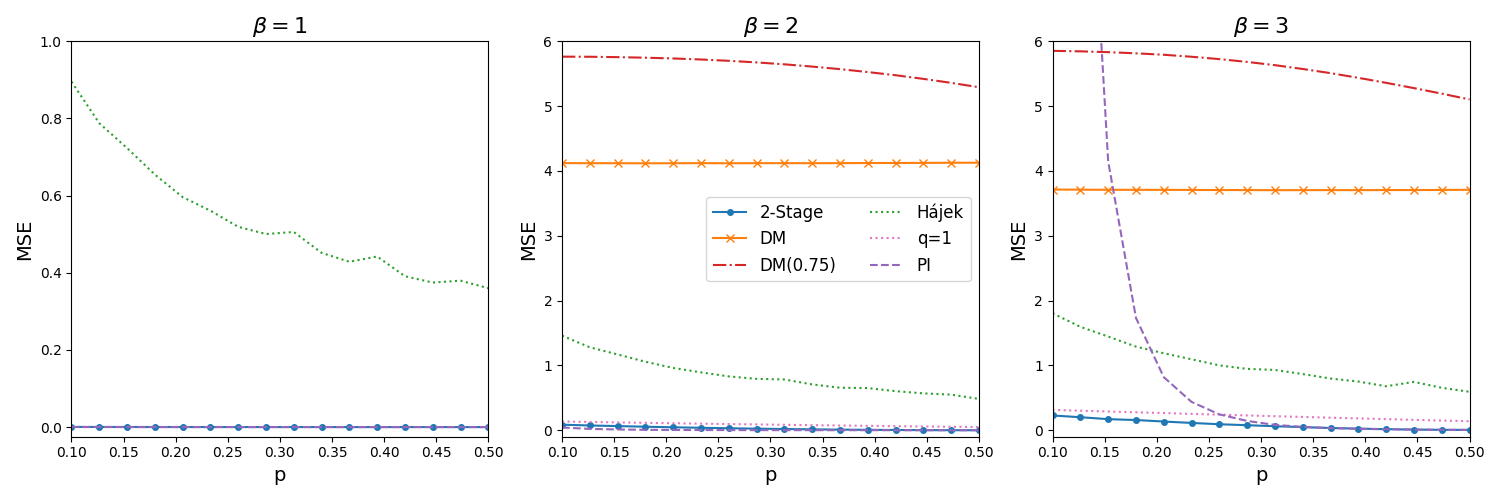}
    \caption{\textsc{Amazon} Network. MSE of different estimators as a function of treatment budget $p$.}
    \label{fig:amazon_est_mse}
\end{figure} 

\begin{figure}[h]
\centering
    \begin{subfigure}[b]{0.9\textwidth}
        \centering
        \includegraphics[width=0.9\linewidth]{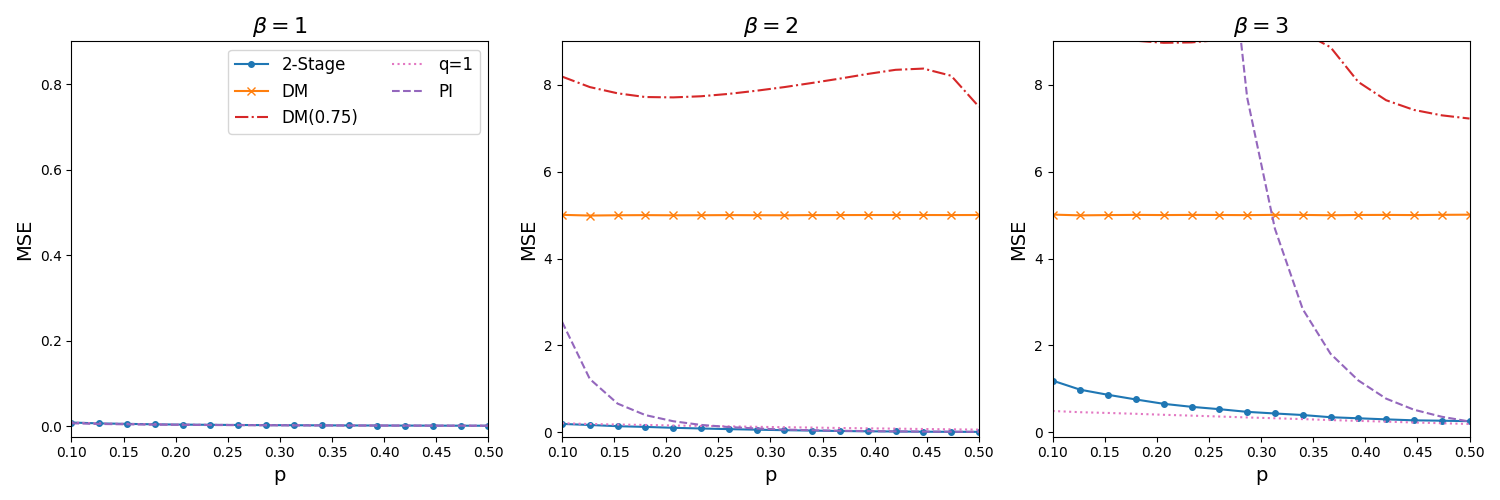}
        \caption{MSE of different estimators as a function of treatment budget $p$. }
        \label{fig:blog_est_mse}
    \end{subfigure}
    \begin{subfigure}[b]{0.9\textwidth}
        \centering
        \includegraphics[width=0.9\linewidth]{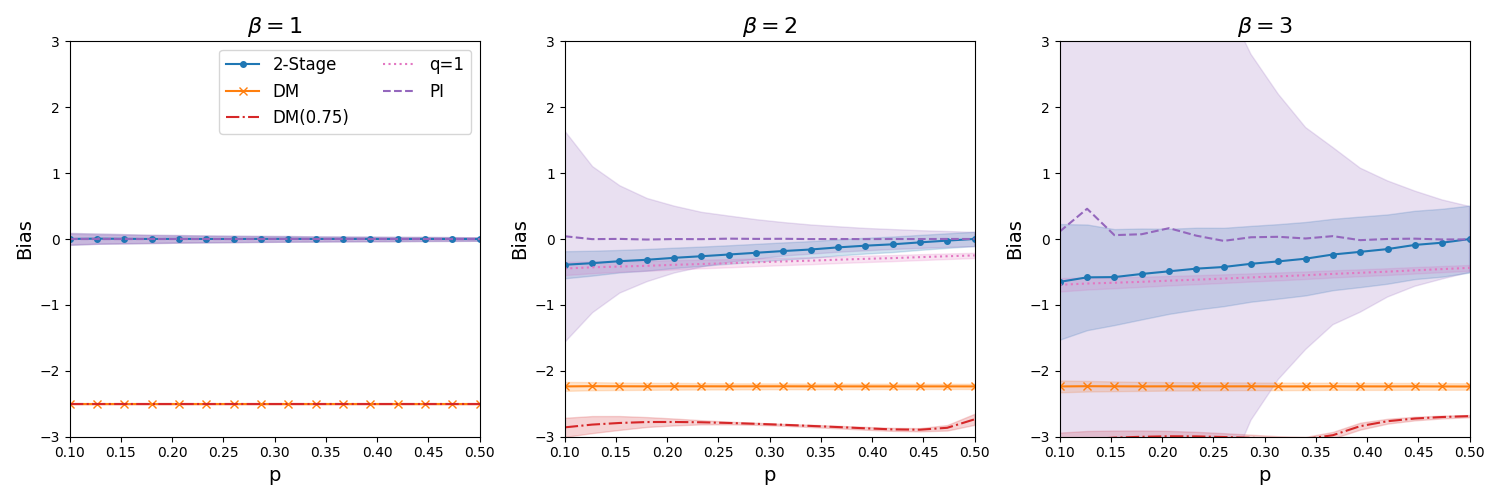}
        \caption{Bias and variance of different estimators as a function of treatment budget $p$.}
        \label{fig:blog_est}
    \end{subfigure}
    \caption{\textsc{BlogCatalog} Network.} \label{fig:blog}
\end{figure}

In Figure \ref{fig:blog}, we show the MSE and the bias and standard deviation of the different estimators under the \textsc{BlogCatalog} network. We omit the \textsf{H\'{a}jek} estimator because the network degree is very high; under unit randomization, the estimator is often undefined. In all cases, the two difference-in-means estimators are very biased, so their MSE is much worse than the PI estimators. Similar to the \textsc{Amazon} network, when $\beta=1$ the PI estimators are almost indistinguishable, with \textsf{2-Stage} coming in with slightly higher MSE due to a small increase in variance. The difference-in-means estimators do not appear on the plot since their MSE exceeds the plotting range. When $\beta=2$, we see that \textsf{PI} has higher MSE for smaller values of $p$ due to an increased variance (the estimator is unbiased). \textsf{2-Stage} has much lower variance than \textsf{PI} and its bias decreases as $p$ approaches $0.5$. \textsf{q=1} has even lower variance than \textsf{2-Stage} and similar bias, but its bias remains worse than \textsf{2-Stage}. However, their MSE remains comparable. When $\beta=3$, there is a much clearer difference between the PI estimators. For most values of $p$, \textsf{PI} has the worst variance (due to extrapolation with a richer model), while \textsf{q=1} has the worst bias (due to the heavier reliance on the subsampling). Meanwhile, \textsf{2-Stage} is in between, with lower bias, but larger variance than \textsf{q=1} and with higher bias, but lower variance than \textsf{PI}. In terms of MSE, \textsf{q=1} outperforms the other estimators. Recall that the \textsf{q=1} estimator is equivalent to the $\beta=1$ version of the one-stage PI estimator. This setting shows that although the two-stage approach can greatly reduce error over the one-stage approach even without clustering, an even simpler design (one-stage rollout over just two time steps) and estimator (using observations from the two time steps) can still outperform.

\begin{figure}[h]
\centering
    \begin{subfigure}[b]{0.9\textwidth}
        \centering
        \includegraphics[width=0.9\linewidth]{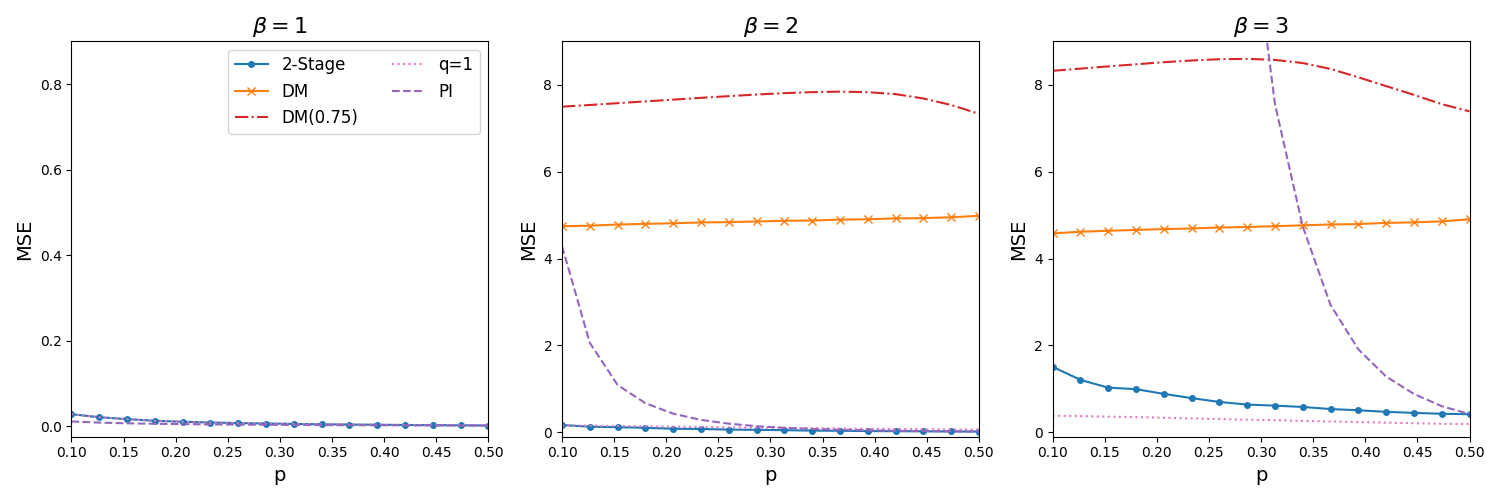}
        \caption{MSE of different estimators as a function of treatment budget $p$. }
        \label{fig:email_est_mse}
    \end{subfigure}
    \begin{subfigure}[b]{0.9\textwidth}
        \centering
        \includegraphics[width=0.9\linewidth]{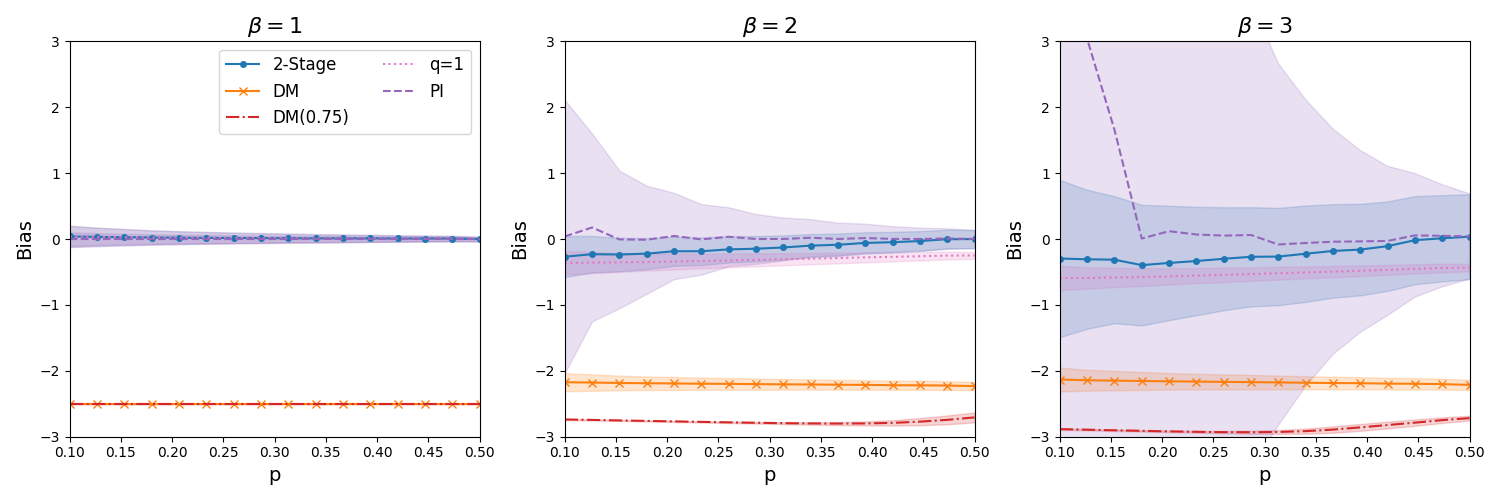}
        \caption{Bias and variance of different estimators as a function of treatment budget $p$. }
        \label{fig:email_est}
    \end{subfigure}
    \caption{\textsc{Email} Network.} \label{fig:email}
\end{figure}

The results from Figure \ref{fig:email} are the same as Figure \ref{fig:blog}.

\subsection{Additional Experiments: Comparing Different Clusterings}
We compare the performance of the \textsf{2-Stage} estimator under two clustering methods versus no clustering in the first stage of the experimental design. In the \textbf{clustering with full graph knowledge}, we cluster the true underlying graph using the \textsc{METIS} clustering library \cite{karypis1998fast}. In the \textbf{clustering with covariate knowledge}, clusters are based on features. When each vertex is assigned to one feature, we use these assignments as the clustering. When vertices may have multiple features we form a feature graph --- a weighted graph, where the weight of edge $(i,j)$ is the number of feature labels shared by $i$ and $j$ --- and cluster this feature graph using \textsc{METIS}. In all plots, the column faceting indicates the type of clustering and the $y$-axis varies $q$ on the interval $[p,1]$, where $p=0.1$. 

We also include tables with various pieces of information pertaining to the performance of the two-stage design and estimator, including clustering metrics such as number of cut edges, the cut effect  $C(\delta(\Pi))$, and the empirical variance across clusters of cluster average influences $\widehat{\Var}\big(\bar{L}_\pi\big)$. The latter two metrics are defined in Section \ref{sec:newmethod}. In each row, $q_{\min}$ is the value of $q$ that minimizes the MSE and the column MSE($q_{\min}$) contains that value.

\begin{figure}[h]
    \centering
    \includegraphics[width=\linewidth]{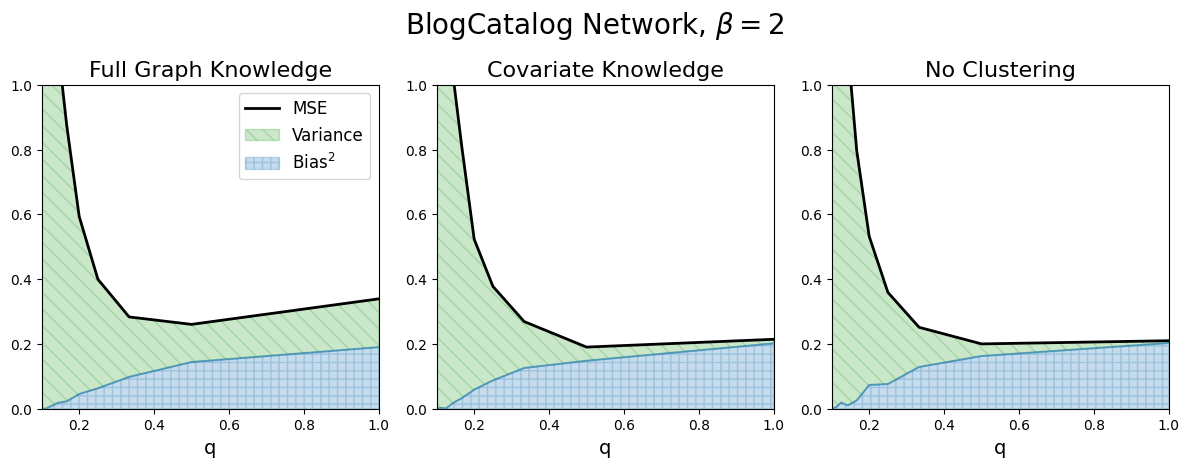}
    \caption{\textsc{BlogCatalog} Network. MSE of \textsf{2-Stage} estimator under two different clusterings versus no clustering, as a function of $q$.}
    \label{fig:blog-clusterings}
\end{figure}

\begin{table}[h]
    \centering
    \caption{Clustering Metrics for \textsc{BlogCatalog} Network}
    \label{tab:clustering_metrics_blog}
    
    \begin{tabular}{c|c|c|c|c|c|c}
        $\beta$ & Cluster & $\widehat{\Var}\big(\bar{L}_\pi\big)$ & $C(\delta(\Pi))$ & Cut Edges & $q_{\min}$ & MSE($q_{\min}$) \\
        \hline
        $2$ & Full & 0.697 & 0.471 & 604504 & 0.5 & 0.260\\
        $2$ & Covariate & 0.059 & 0.486 & 643080 & 0.5 & 0.190\\
        $3$ & Full & 0.703 & 0.717 & 604504 & 1 & 0.610\\
        $3$ & Covariate & 0.060 & 0.734 & 643080 & 1 & 0.486 
    \end{tabular}
\end{table}

In Figure \ref{fig:blog-clusterings}, we show the results for the \textsc{BlogCatalog} network under a model with degree $\beta=2$. In this case, the clusterings each have $n_c=50$ clusters. For this network, clustering does not appear to be of any help in reducing the bias and at worst, under a clustering that uses full graph knowledge, increases variance. Taking a look at the first two rows of Table \ref{tab:clustering_metrics_blog} sheds some light on this. We see that the $\widehat{\Var}\big(\bar{L}_\pi\big)$ under a clustering that uses full graph knowledge is more than ten times higher than a clustering that only uses covariate information. The number of cut edges is similar under both clusterings and thus so is the cut effect.  
In this example, it would appear that one is better off not clustering at all. 

In Figure \ref{fig:email-clusterings}, we show the results for the \textsc{Email} network under a model with degree $\beta=2$. In this case, the clusterings each have $n_c=42$ clusters. We see an advantage with clustering on covariates in particular. The highest bias, but lowest variance, is under no clustering. The clustering with full knowledge certainly decreases variance versus no clustering, but at the expense of incurring a lot of variance. The lowest MSE is achieved by the covariate knowledge clustering at $q=1$, which strikes a balance between bias and variance. Taking a look at Table \ref{tab:clustering_metrics_email}, we see that the $\widehat{\Var}\big(\bar{L}_\pi\big)$ term is similar under both clusterings. However, the covariate clustering cuts about a quarter less edges than the full knowledge clustering and thus has a smaller cut effect. 

\begin{figure}[h!]
    \centering
    \includegraphics[width=\linewidth]{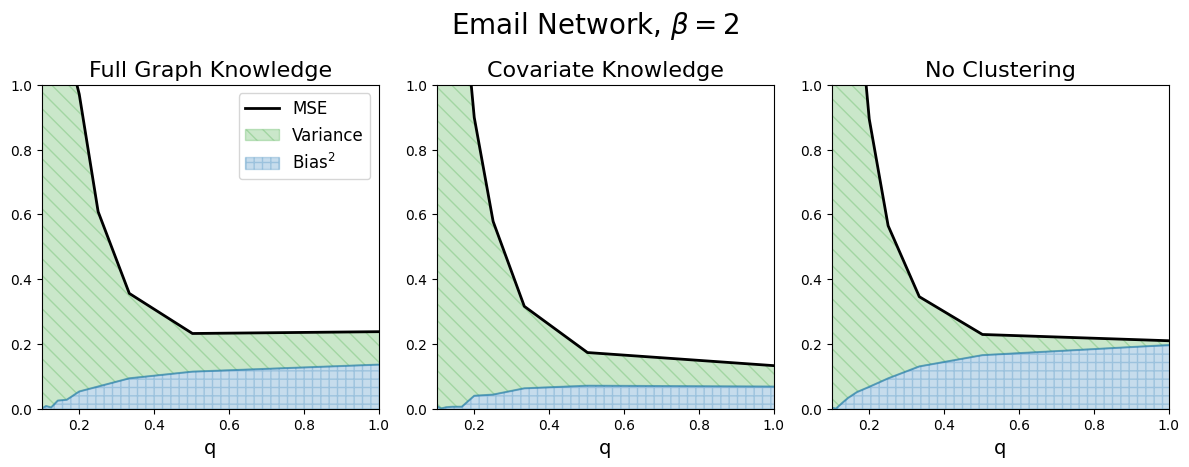}
    \caption{\textsc{Email} Network. MSE of \textsf{2-Stage} estimator under two different clusterings versus no clustering, as a function of $q$.}
    \label{fig:email-clusterings}
\end{figure}

\begin{table}[h]
    \centering
    \caption{Clustering Metrics for Email Network}
    \label{tab:clustering_metrics_email}
    
    \begin{tabular}{c|c|c|c|c|c|c}
        $\beta$ & Cluster & $\widehat{\Var}\big(\bar{L}_\pi\big)$ & $C(\delta(\Pi))$ & Cut Edges & $q_{\min}$ & MSE($q_{\min}$) \\
        \hline
        $2$ &Full & 0.399 & 0.442 & 21756 & 0.5 & 0.232\\
        $2$ & Covariate & 0.398 & 0.372 & 16284 & 1 & 0.133 \\
        $3$ &Full & 0.417 & 0.686 & 21756 & 1 & 0.483\\
        $3$ & Covariate & 0.412 & 0.591 & 16284 & 1 & 0.288
    \end{tabular}
\end{table}

\subsection{Additional Experiments: Homophily Parameter $b=0.5$.}
In this section, we show some results when the model exhibits homophily by setting the parameter $b=0.5$. All other parameters are set to the same values as previous plots. Although there are some small visual differences between the plots in this section and the plots throughout the rest of this work, the analyses and conclusions remain the same.
For example, we can compare Figure \ref{fig:amazon_est_mse} (where $b=0$) with \ref{fig:amazon_est_homophily} (where $b=0.5$). Both of these show the MSE of different estimators for different values of treatment budgets $p$ and different model degrees $\beta$. Notice the difference in the scaling on the $y$-axis, particularly for $\beta=2$ and $\beta=3$. However, the patterns are the same: for most values of $p$, the two difference in means estimators have the highest MSE, followed by the H\'{a}jek estimator, and then followed by the three PI estimators. When $\beta=3$, the MSE of the vanilla PI estimator is extremely high for small values of $p$, but gets smaller than the non-PI estimators around $p=0.2$. In general, for $\beta=3$, \textsf{2-Stage} tends to outperform \textsf{PI} for many paramter values and for some networks, \textsf{q=1} has the smallest MSE in some cases. When $\beta=2$, we see that the two stage approach improves over the one stage approach under the \textsc{BlogCatalog} and \textsc{Email} networks for small values of $p$. When $\beta=1$ the performances of the PI estimators are similar, with \textsf{2-Stage} performing ever so slightly worse. 

In Figure \ref{fig:compare_clusterings_homophily}, we show the performance of the two stage approach under two different clustering versus no clustering for three different networks. The \textsc{BlogCatalog} and \textsc{Email} network results are very similar to those in Figures \ref{fig:blog-clusterings} and \ref{fig:email-clusterings}. The Amazon network result sheds some light onto why the scaling is different in the Amazon MSE plots: the bias has a larger magnitude. Part of this is likely attributable to the fact that switching from $b=0$ to $b=0.5$ changed the magnitude of the baseline outcomes, and therefore all outcomes.

\begin{figure}[h]
    \centering
    \begin{subfigure}[b]{\textwidth}
        \centering
        \includegraphics[width=\linewidth]{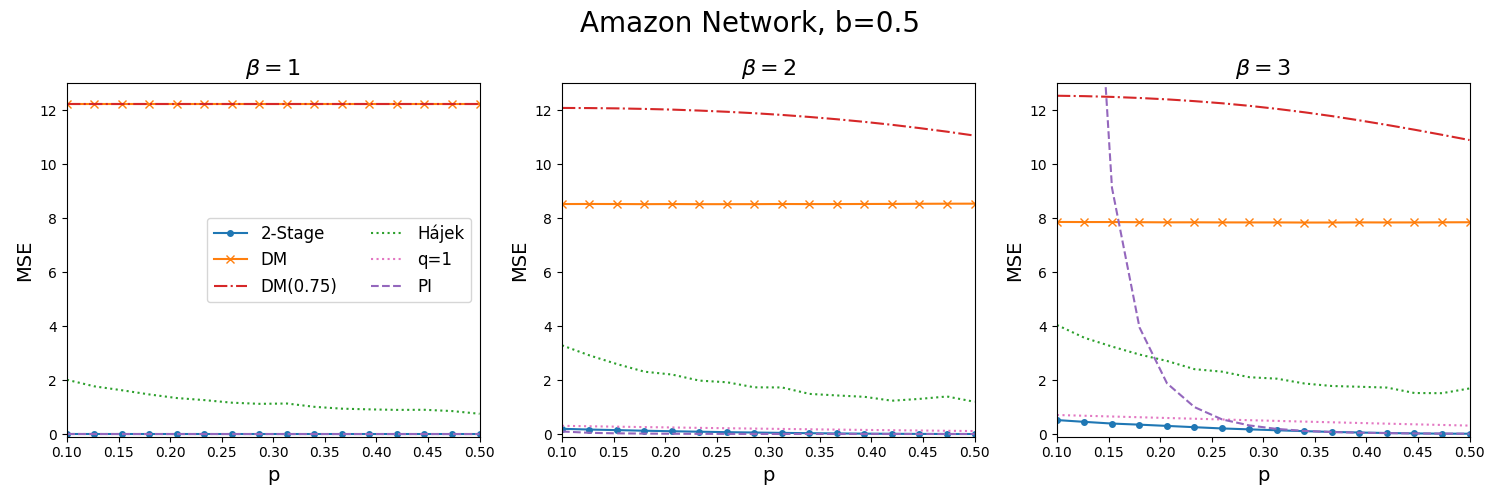}
        \caption{}
        \label{fig:amazon_est_homophily}
    \end{subfigure}
    \begin{subfigure}[b]{\textwidth}
        \centering
        \includegraphics[width=\linewidth]{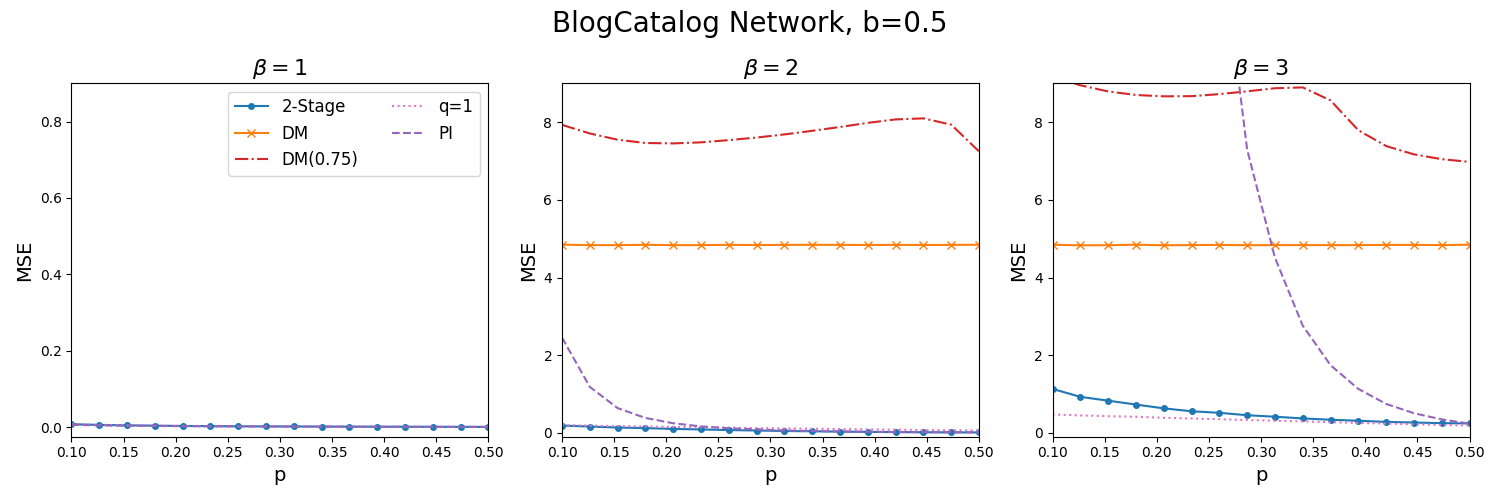}
        \caption{}
    \end{subfigure}
    \begin{subfigure}[b]{\textwidth}
        \centering
        \includegraphics[width=\linewidth]{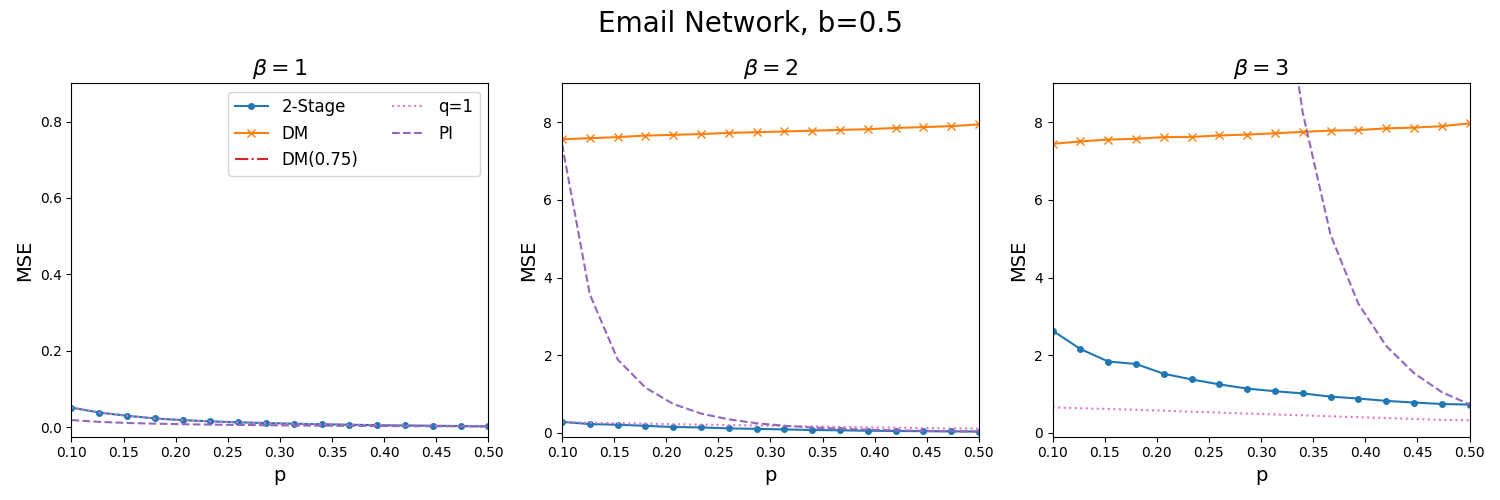}
        \caption{}
    \end{subfigure}
    \caption{} \label{fig:compare_estimators_homophily}
\end{figure}

\begin{figure}[h]
    \centering
    \begin{subfigure}[b]{\textwidth}
        \centering
        \includegraphics[width=0.9\linewidth]{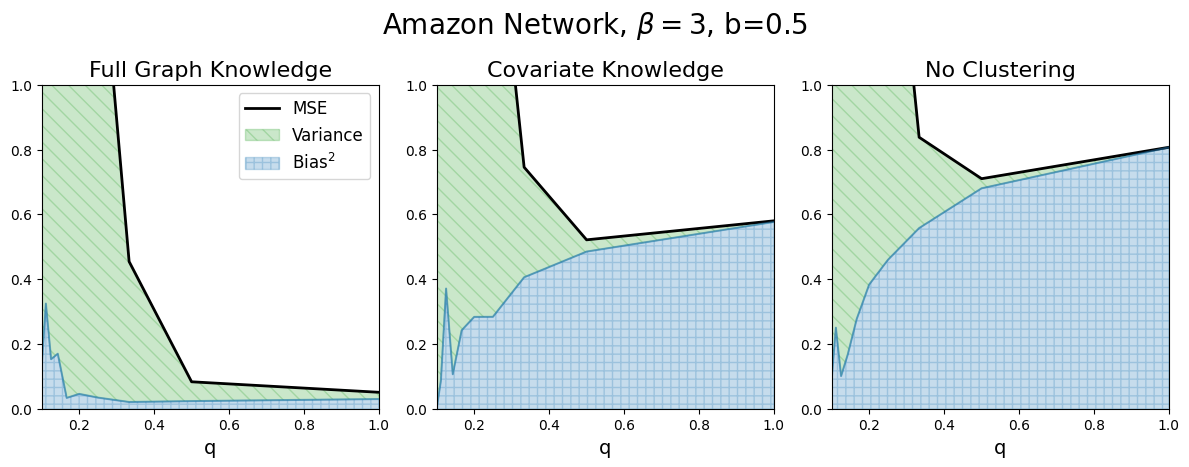}
        \caption{}
    \end{subfigure}
    \begin{subfigure}[b]{\textwidth}
        \centering
        \includegraphics[width=0.9\linewidth]{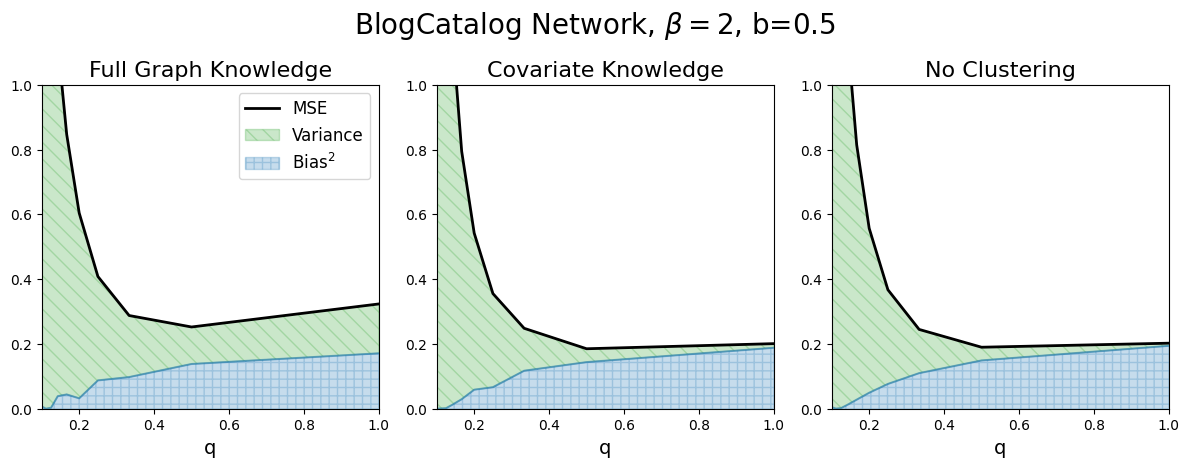}
        \caption{}
    \end{subfigure}
    \begin{subfigure}[b]{\textwidth}
        \centering
        \includegraphics[width=0.9\linewidth]{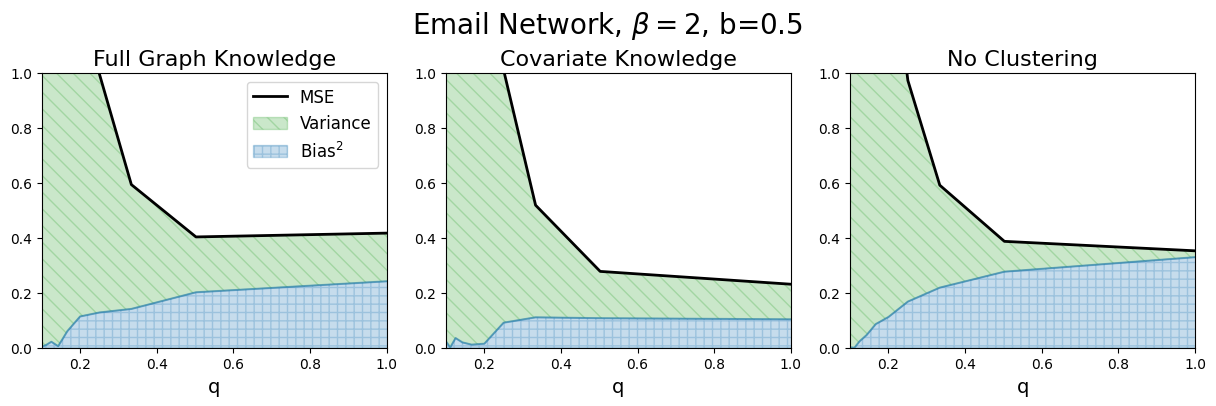}
        \caption{}
    \end{subfigure}
    \caption{} \label{fig:compare_clusterings_homophily}
\end{figure}

\end{document}